\documentclass[11pt]{article}

\usepackage[english]{babel}
\usepackage[T1]{fontenc}
\usepackage[margin=1in]{geometry} 
\usepackage{xcolor}

\usepackage[]{algorithm2e}

\usepackage{amsmath,amssymb,amsthm}
\usepackage{graphicx}
\usepackage[round]{natbib}
\usepackage{comment}
\usepackage{hyperref}

\usepackage[capitalize]{cleveref}
\Crefname{appsec}{Appendix}{Appendices}

\hypersetup{
    colorlinks=true,
    linkcolor= blue,
    filecolor=blue,
    citecolor = blue,      
    urlcolor=cyan,
    }

\usepackage[capitalize]{cleveref}

\usepackage{thmtools}
\usepackage{thm-restate}

\newtheorem{theorem}{Theorem}
\newtheorem{lemma}[theorem]{Lemma}

\newtheorem{corollary}[theorem]{Corollary}

\newcommand{\TLP}{\textup{LP1}}
\newcommand{\LPmatroid}{\textup{LP2}}

\newcommand{\HLP}{\textup{LP-H}}

\newcommand{\maximize}{\textbf{maximize:}}
\newcommand{\minimize}{\textbf{minimize:}}

\newcommand{\subjectto}{\textbf{s.t.:}}

\newcommand{\st}{ : }

\newcommand{\opt}{\operatorname{OPT}}

\title{Assigning and Scheduling Generalized Malleable Jobs under\\ Subadditive or Submodular Processing Speeds
}

\author{Dimitris Fotakis \\ 
National Technical University of Athens \\ \texttt{fotakis@cs.ntua.gr}
\and Jannik Matuschke \\ KU Leuven \\ \texttt{jannik.matuschke@kuleuven.be}
\and Orestis Papadigenopoulos \\ The University of Texas at Austin \\ \texttt{papadig@cs.utexas.edu}
}

\begin{document}

\maketitle

\begin{abstract}
Malleable scheduling is a model that captures the possibility of parallelization to expedite the completion of time-critical tasks.
A malleable job can be allocated and processed simultaneously on multiple machines, occupying the same time interval on all these machines. 
We study a general version of this setting, in which the functions determining the joint processing speed of machines for a given job follow different discrete concavity assumptions (subadditivity, fractional subadditivity, submodularity, and matroid ranks). We show that under these assumptions the problem of scheduling malleable jobs at minimum makespan can be approximated by a considerably simpler assignment problem.
Moreover, we provide efficient approximation algorithms for both the scheduling and the assignment problem, with increasingly stronger guarantees for increasingly stronger concavity assumptions, including a logarithmic approximation factor for the case of submodular processing speeds and a constant approximation factor when processing speeds are determined by matroid rank functions.
Computational experiments indicate that our algorithms outperform the theoretical worst-case guarantees.
\end{abstract}

\section{Introduction}
\label{sec:intro}

The use of {\em parallelization}, i.e., the simultaneous processing of a job on multiple machines to expedite its completion, is wide-spread in task scheduling systems, where careful planning is crucial to reduce the overall makespan of time-critical schedules. 
In computing, the advent of massive parallelism has enabled or significantly facilitated a diverse range of applications including 
Cholesky factorization~\citep{dongarra1998numerical}, 
molecular simulation~\citep{bernard1999large}, 
web search index update~\citep{wu2015algorithms}, 
or training neural networks~\citep{fujiwara2018effectiveness}. 
Also in the context of production and logistics, parallelism is used in a wide variety of operational task scheduling applications, such as
quay crane allocation in naval logistics~\citep{imai2008simultaneous,blazewicz2011berth}, workforce assignment in production~\citep{dolgui2018optimal}, recharging electrical vehicles~\citep{nguyen2018electrical}, textile production~\citep{serafini1996scheduling,mourtos2021scheduling}, disaster relief operations~\citep{vanderster2010allocation}, or cleaning activities on trains~\citep{bartolini2017scheduling}. 

A frequently used model for capturing parallelization in scheduling, including many of the aforementioned examples, is that of 
\emph{malleable} or \emph{moldable} jobs~\citep{du1989complexity}. At a conceptual level, a malleable job can be assigned to an arbitrary subset of the available machines to be processed \emph{non-preemptively} and in \emph{unison}, i.e., with the same starting and completion time on each of the allocated machines. Importantly, the malleable scheduling model allows the scheduler to decide on the degree of parallelization for each individual job by choosing the set of machines to which it is assigned, in contrast to non-malleable parallel machine models in which each task is assigned to a single machine.

In this work, we study a new generalized model of malleable scheduling that significantly goes beyond the identical machine setting that has been predominantly studied in literature before. We show that, under natural discrete concavity assumptions on the functions determining the processing speed of machine sets (subadditivity, fractional subadditivity, submodularity, matroid ranks), increasingly strong approximation results can be obtained for this generalized model. 
Our first main result shows how the task of finding a \emph{schedule} of small makespan can be approximated by the significantly simpler task of finding an \emph{assignment} of jobs to machine sets.
We further provide efficient approximation algorithms for this latter assignment problem, which via the first result also translate into approximation algorithms for the scheduling problem (see \cref{table:summary} for a summary of the resulting approximation factors). Finally, we demonstrate the applicability of our approach in a computational study for this latter heuristic and the aforementioned assignment-to-schedule transformation.

\subsection{The Malleable Scheduling Model}
Consider a set of \emph{jobs} $J$ to be assigned on a set of \emph{machines} $M$.  Each job $j \in J$ is equipped with a \emph{processing time function} $f_j : 2^M \rightarrow \mathbb{R}_{\geq 0}$ that specifies the time $f_j(S)$ needed for the completion of $j$ when assigned to a subset of machines $S \subseteq M$.
A \emph{schedule} consists of an \emph{assignment} $\mathbf{S} = (S_j)_{j \in J}$ that specifies a non-empty set of machines $S_j \subseteq M$ for each job $j \in J$, and of a \emph{starting time} vector $\mathbf{t} = (t_j)_{j \in J}$.
In such a schedule $(\mathbf{S}, \mathbf{t})$, the machines in $S_j$ all jointly process job $j$ for the same interval of time, starting at time $t_j$ and finishing at time $t_j + f_j(S_j)$. At any given moment, each machine can only process a single job, i.e., there are no jobs $j, j' \in J$ with $S_j \cap S_{j'} \neq \emptyset$ and $t_j < t_{j'} < t_{j} + f_j(S_j)$.
The \emph{makespan} of a schedule $(\mathbf{S}, \mathbf{t})$ is $C_{\max} := \max_{j \in J} t_j + f_j(S_j)$, i.e., the time when all jobs are completed. 

We are interested in finding schedules of minimum makespan. Throughout this article, we will call this the \textsc{Scheduling} problem. Related to this problem is the \textsc{Assignment} problem, which asks for an assignment $\mathbf{S}$ minimizing the \emph{maximum load} $L(\mathbf{S}) := \max_{i \in M} \sum_{j \in J: i \in S_j} f_j(S_j)$ over all machines. Note that the latter problem does not include information on the starting time of jobs and that the maximum load in any assignment $\mathbf{S}$ is a lower bound on the makespan of any schedule $(\mathbf{S}, \mathbf{t}$) using that same assignment. 
We will discuss the relation between \textsc{Scheduling} and \textsc{Assignment} further in \cref{sec:challenges} below.

Rather surprisingly, and despite the significant theoretical and practical interest in the model, most of the work on scheduling malleable jobs considers the case of \emph{identical} machines. Note that in this case, the processing time of a job only depends on the \emph{number} of allocated machines, i.e., $f_j(S) = \bar{f}_j(|S|)$ for some function $\bar{f}_j : \mathbb{Z}_{\geq 0} \rightarrow \mathbb{R}_{\geq 0}$.
In this setting, it is commonly assumed that the processing time $\bar{f}_j(k)$ is non-increasing in the number of machines $k$, whereas the product $k \cdot \bar{f}_j(k)$, i.e., the number of machine hours needed to complete the job, is non-decreasing. The latter assumption is known as the \emph{monotone work assumption}. 
It accounts for the fact that, in most practical applications, the amount of work needed to complete a job does not decrease when using parallelization (rather, parallelization often incurs a communication/coordination overhead) and that the combined processing speed of a set of machines is not larger than the sum of their individual processing speeds, i.e., no synergies are obtained from combining multiple machines.

The identical machine case is very well understood from an optimization perspective. Constant-factor approximation algorithms for the problem are known since the work of~\citet{turek1992approximate}.
A line of work establishing improved approximation results \citep{mounie1999efficient,jansen2002linear,mounie2007frac32} recently culminated in a polynomial-time approximation scheme, implicit in the combined work of \citet{jansen2010approximation} and \citet{jansen2018scheduling}. Thus, the minimum makespan can be approximated efficiently up to arbitrary precision in the identical machine case, matching the complexity of the corresponding non-malleable scheduling problem known as $P||C_{\max}$ \citep{hochbaum1987using}.

Literature on malleable scheduling beyond the identical machine setting is scarce. 
A model similar to malleable scheduling in which unrelated machines have been studied is that of \emph{splittable jobs}. 
In this model, each job can be split into arbitrary fractions that can be distributed on different machines (without any requirement for synchronization). 
The time required for processing a fraction $x_{ij} > 0$ of job $j$ on machine $i$ is given by $s_{ij} + p_{ij}x_{ij}$, where $s_{ij}$ is a fixed setup time and $p_{ij}$ is a processing time factor.
\citet{correa2015strong} provide a $(1 + \varphi)$-approximation algorithm for this setting, where $\varphi \approx 1.62$ is the golden ratio, and show that this factor is tight for a natural LP relaxation of the problem.
In the malleable model, non-identical machines have only recently been studied for the setting of \emph{speed-implementable processing-time functions} \citep{fotakis2020malleable}. In this setting, each machine $i$ is associated with an unrelated ``speed'' $s_{ij}$ for each job $j$, and processing time is a function of the total allocated speed, i.e., $f_j(S) = \bar{f}_j(\sum_{i \in S} s_{ij})$ for some function $\bar{f}_j : \mathbb{R}_{\geq 0} \rightarrow \mathbb{R}_{\geq 0}$. The monotone work assumption naturally extends to this setting by requiring $s \cdot \bar{f}_j(s)$ to be  non-decreasing. \citet{fotakis2020malleable} devised an LP-based $3.16$-approximation for this setting.

A significant limitation inherent to models as the two discussed above, in which the contribution of a machine to a job can be expressed by a single number, is their inability to capture interaction effects caused by different combinations of machines when processing the same job.
Complicated interdependencies may arise, e.g., from combining machines of different types or characteristics, from operational requirements, or from the local topology of the processing platform. 
In computing, e.g., modern heterogeneous parallel systems are typically consist of CPUs, GPUs, and I/O nodes, featuring different architectures and memory restrictions~\citep{bleuse2017scheduling}. 
To make optimal use of such systems, scheduling algorithms need to take the different architectures and the topology of the underlying communication network into consideration~\citep{bampis2020scheduling}. 
Another example is workforce assignment in production, where the assumption of identically-skilled workers is inappropriate in environments characterized by short-term contracts and an absence of standardized training~\citep{nakade2008optimal}, in sheltered work centers for the disabled~\citep{borba2014heuristic}, or in teams combining highly-trained specialized workers~\citep{walter2016minimizing}. 
These and similar situations require models capturing the combinatorial interaction effects of machines or workers assigned to the same task.

\subsection{Generalized Malleable Scheduling with Concave Processing Speeds}
In this work, we investigate a generalized setting in which processing times depend on the actual \emph{set} of machines used for a job. This model not only allows for machines to be distinctively well-suited for different jobs but it also captures combinatorial relations between different machines, such as the incompatibility to jointly process the same job.

To extend the monotone work assumption to the set function context, we define $g_j(S) := 1 / f_j(S)$ for every $j \in J$ and $S \subseteq M$ (as jobs have to be allocated to non-empty machine set, we assume w.l.o.g.\ that $f_j(\emptyset) = \infty$ and $g_j(\emptyset) = 0$ for all $j \in J$). 
Intuitively, $g_j(S)$ corresponds to the joint processing speed with which the machines in $S$ process job $j$, which is represented by one (normalized) unit of work in the numerator.
We thus refer to the function $g_j : 2^M \rightarrow \mathbb{R}_{\geq 0}$ as the \emph{speed function} for job $j \in J$. Unless noted otherwise, we will assume that each $g_j$ is given by a value oracle that, given a set $S$, returns the value $g_j(S)$.

Recall that the monotone work assumption for the identical machine setting implies the absence of synergies in processing speed when allocating more machines to a job. 
To impose a similar assumption on the processing speed functions $g_j$, we consider the following four classes of functions from the field of discrete convexity:
\begin{itemize}
  \item {\bf subadditivity:} A function $g : 2^M \rightarrow \mathbb{R}_{\geq 0}$ is \emph{subadditive} if $g(S \cup T) \leq g(S) + g(T)$ for all $S, T \subseteq M$. 
  \item {\bf fractionally subadditivity:} A function $g : 2^M \rightarrow \mathbb{R}_{\geq 0}$ is \emph{fractionally subadditive} (also known as \emph{XOS} in the literature), if $g(T) \leq \sum_{S \subseteq T} \alpha(S) g(S)$ for all $T \subseteq M$ and for all $\alpha : 2^M \rightarrow [0, 1]$ that fulfill $\sum_{S \subseteq T : i \in S} \alpha(S) \geq 1$ for all $i \in T$.
  \item {\bf submodularity:} A function $g : 2^M \rightarrow \mathbb{R}_{\geq 0}$ is \emph{submodular} if $g(S \cup T) + g(S \cap T) \leq g(S) + g(T)$ for all $S, T \subseteq M$. This definition is equivalent to the following diminishing returns property: $g(S \cup \{i\}) - g(S) \geq g(T \cup \{i\}) - g(T)$ for all $S \subseteq T \subseteq M$ and $i \in M \setminus T$.
  \item {\bf matroid ranks:} A \emph{matroid} on the ground set $M$ is a non-empty set family $\mathcal{F} \subseteq 2^M$ such that (i) $T \in \mathcal{F}$ implies $S \in \mathcal{F}$ for all $S \subseteq T$ and (ii) for every $S, T \in \mathcal{F}$ with $|S| < |T|$ there is an $i \in T \setminus S$ such that $S \cup \{i\} \in \mathcal{F}$. The \emph{rank function} of a matroid $\mathcal{F}$ is defined by $r(S) := \max_{S \subseteq T : S \in \mathcal{F}} |S|$ for $S \subseteq M$.
  A function $g : 2^M \rightarrow \mathbb{R}_{\geq 0}$ is a \emph{scaled matroid rank} if $g(S) = \alpha r(S)$ for all $S \subseteq M$, where $\alpha > 0$ and $r$ is the rank function of some matroid.
\end{itemize}

Each of these classes captures the concepts of ``no synergies'' or ``diminishing returns'' for set functions with an increasing degree of restrictiveness (i.e., subadditive functions constitute a superclass of fractionally subadditive functions, which form a superclass of submodular functions etc.).
All four classes have been widely studied, receiving particular attention in operations research and economics literature for their applications in diverse fields such as auction markets~\citep{candogan2015iterative}, cooperative game theory~\citep{bondareva1963some,feige2009maximizing},
inventory management~\citep{levi2014logic}, machine learning~\citep{bach2013learning}, or queuing theory~\citep{ahn2013flexible}; also see the survey by \citet{iwata2008submodular} and the textbook by \citet{murota2003discrete}.

We remark that it is natural to assume processing speeds to be monotone (i.e., $S \subseteq T$ implies $g(S) \leq g(T)$). However, most of our results work without this assumption, and we will specifically point out when monotonicity is required.
Before we elaborate our results for the generalized malleable scheduling model under the aforementioned discrete concavity assumptions, we highlight some computational and practical challenges inherent to malleable scheduling, and the resulting distinction of assigning and scheduling.

\subsection{Assigning vs.\ Scheduling and Practical Challenges}
\label{sec:challenges}
A particular challenge in malleable scheduling, both in terms of computation and practical implementation of schedules, is the need to coordinate schedules on different machines to allow  jobs to be processed in unison.
In particular, this feature of malleable scheduling motivates the distinction between \emph{assignment} and \emph{schedule}.

Recall that the \textsc{Assignment} problem asks for an assignment minimizing the maximum machine load (ignoring starting times), whereas the \textsc{Scheduling} problem asks for a schedule of minimum makespan (including starting times).
In classic non-malleable scheduling these two problems are equivalent, as any assignment can be turned into a schedule whose makespan equals the maximum machine load, just by letting each machine process the jobs assigned to it in arbitrary order.
This is no longer true in malleable scheduling, as starting times for jobs cannot be determined independently on each machine. 
Indeed, it is easy to construct examples of assignments $\mathbf{S}$ such that any schedule $(\mathbf{S}, \mathbf{t})$ using this assignment has a makespan of at least $|J|/2$ times the maximum machine load, independent from assumptions on the processing speed functions; see \cref{app:assignment-scheduling-gap}.

From a computational point of view, the need for synchronization seems to make finding good (i.e., low-makepsan) schedules considerably harder than the task of finding good assignments. Indeed, the \textsc{Assignment} problem has a natural and well-structured IP formulation---based on binary variables $x(S, j)$ indicating the set of machines $S$ assigned to job $j$---that is amenable to column generation approaches; see \cref{app:assignment-ip} for details. 
Straightforward IP formulations for the \textsc{Scheduling} problem, on the other hand, such as time-indexed formulations, not only incur a considerable blow-up in number of variables and constraints, but also appear to be much less well-structured.

Finally, executing schedules that rely heavily on parallelization requires frequent regrouping and well-organized synchronization of sets of machines to process different jobs in unison, creating a managerial challenge.
In particular, such schedules are susceptible to delays that can easily propagate throughout the schedule. In fact,
due to the interdependence of schedules on different machines, a delay in a single job can make it necessary to reschedule jobs on all machines, even those not involved in processing the delayed job; see \cref{fig:delay}.
Similarly, small delays of multiple jobs, even when processed on disjoint sets of machines, can accumulate to a significant extension of the makespan.

\begin{figure}[t]
\centering
\begin{minipage}{.5\textwidth}
  \centering
  \includegraphics[width=.9\linewidth]{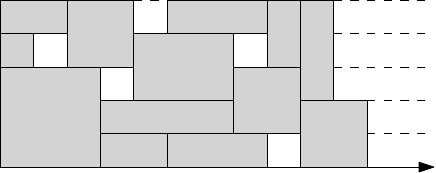}
  
  (a)
  
\end{minipage}%
\begin{minipage}{.5\textwidth}
  \centering
  \includegraphics[width=.9\linewidth]{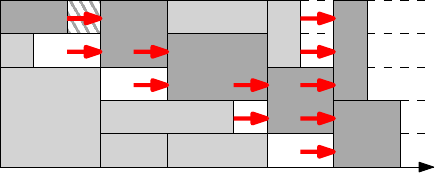}
  
  (b)
  
\end{minipage}
  \caption{(a) Original schedule ($x$-axis is the time and $y$-axis the set of machines); (b) Delay propagating through a schedule: A unit-time delay in the upper-left job causes a cascade of delays.}
  \label{fig:delay}
\end{figure}

\subsection{Our Results}

In the following, we give an overview of our results and point to the corresponding sections where they are discussed in detail.

\subsubsection{Turning Assignments Into Schedules}

Our first result is a transformation that turns any given assignment into a schedule (with a slightly modified assignment), addressing both the computational and practical challenges described in \cref{sec:challenges}. We show that the resulting schedule incurs only a bounded increase in makespan compared to the original assignment. Thus, the \textsc{Scheduling} problem can be approximated via the simpler \textsc{Assignment} problem. The resulting schedules moreover have a beneficial structure that significantly facilitates their implementation: only the first job processed on any machine is processed across multiple machines, all later jobs are handled exclusively by a single machine; see \cref{fig:wellstructured}. This makes the schedule easy to execute and robust, as delays in individual jobs can only effect other jobs on the same machine but cannot propagate through the schedule.

For fractionally subadditive speed functions, the increase in makespan is bounded by a factor of $2e/(e-1) \approx 3.16$ in the worst-case.
This carries over to subadditive functions at a logarithmic (in the number of machines) loss via a well-known pointwise approximation result.
Our transformation is based on the Bondareva-Shapley theorem \citep{bondareva1963some}, a result from cooperative game theory that guarantees the existence of the solution to an inequality system defined by a fractionally subadditive function.
While, in general, finding such a solution may require to solve an LP of exponential size, it can be found by a simple greedy algorithm in the case of submodular functions.
Thus, the entire procedure can be carried out in polynomial time for submodular processing speeds (with the same worst-case guarantee of $3.16$).
These results are discussed in detail in \cref{sec:transformation}.
Moreover, in \cref{sec:computational}, we report on the results of computational experiments that demonstrate that applying the transformation in practice results in a much smaller loss than the theoretical worst-case guarantee (approximately $1.17$ on average and in no case larger than $2$), making it possible to find near-optimal solutions to \textsc{Scheduling} via the simpler \textsc{Assignment} problem.

\begin{figure}[t]
\centering
\includegraphics[width=.45\linewidth]{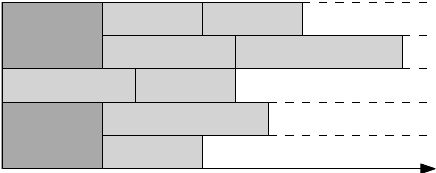}
\caption{Well-structured assignment/ feasible schedule: Each machine processes at most one job in combination with other machines ($x$-axis is the time and $y$-axis the set of machines).}
\label{fig:wellstructured}
\end{figure}

\subsubsection{Approximating the Assignment Problem}
In the second part of the paper, we turn our attention to devising efficient approximation algorithms for both the \textsc{Assignment} and \textsc{Scheduling} problem.
For the special case where processing speeds are defined by scaled matroid ranks, we show that a $4$-approximation for \textsc{Assignment} can be obtained by making use of an appropriate intersection of two matroids.
Using insights from graph orientation, this can be turned into a $5$-approximation algorithm for \textsc{Scheduling}. These results are discussed in \cref{sec:matroids}.

In \cref{sec:submodular}, we present a greedy heuristic for approximating both \textsc{Assignment} and \textsc{Scheduling}, which applies to various classes of speed functions ranging from subadditive to submodular. 
The algorithm runs in phases, where, in each phase, a constant fraction of not-yet-assigned jobs is scheduled within a multiplicative factor of the target makespan. 
To achieve this, our heuristic makes critical use of the structural insights derived in the earlier transformation result, concurrently executing two assignment subroutines based on the \textsc{Generalized Assignment} problem~\citep{shmoys1993approximation} and \textsc{Welfare Maximization} for the class of speed functions at hand~\citep{schrijver03,feige2009maximizing}, respectively.

The algorithm gives a polylogarithmic and a logarithmic approximation guarantee for subadditive and fractionally subadditive processing speeds, respectively, assuming access to a demand oracle. Further, the algorithm can be implemented efficiently using a value oracle for the case of (not necessarily monotone) submodular functions. As we show, access to a stronger demand oracle is necessary in the former case, as any algorithm obtaining a non-trivial approximation factor for \textsc{Assignment} with fractionally subadditive processing speeds given by a value oracle would require an exponential number of queries to this oracle. Finally, we provide improved approximation guarantees for the case of identical jobs, namely, a logarithmic approximation for subadditive and a constant approximation for fractionally subadditive speed functions. In \Cref{table:summary} we present a summary of our results. 

In \cref{sec:computational} we further discuss computational results on the empirical performance of our heuristic for submodular processing speeds, showing that it significantly outperforms its worst-case guarantees for a number of problem instances.

\begin{table}[t]
\small
\centering
\begin{tabular}{ |c|c|c|c| }  
 \hline
 Speed Functions & Transformation & Approximation & Approximation \\ 
  & Assignment $\rightarrow$ Schedule & (Distinct Jobs) & (Identical Jobs)  \\ 
 \hline 
 Subadditive & $\mathcal{O}(\log m)^{a}$ & $\mathcal{O}(\log m \; \log\min\{n,m\})^{b,c}$ & $\mathcal{O}(\log m)^{b,c}$ \\ 
 Fractionally Subadditive & $3.16^{a}$ & $\mathcal{O}(\log \min\{n,m\})^{b,c}$ & $\mathcal{O}(1)^{b,c}$ \\ 
 Submodular & $3.16$  & $\mathcal{O}(\log\min\{n,m\})$ & $\mathcal{O}(1)$  \\ 
 Scaled Matroid Ranks & $3.16$ & $4$ & $2$ \\
 \hline
\end{tabular}
\caption{Summary of results. For each type of speed function, we report the maximum increase in makespan when transforming an assignment into as schedule (\cref{sec:transformation}) and the factor of the approximation algorithms for the \textsc{Assignment} problem with distinct and identical jobs, respectively (\cref{sec:matroids,sec:submodular}). We denote by $n$ and $m$ the number of jobs and machines, respectively. Results marked with $^{a}$ are algorithmic but without polynomial run-time; results marked with $^{b}$ require monotonicity; results marked with $^{c}$ require a demand oracle instead of a value oracle.}
\label{table:summary}
\end{table}

\section{Transforming Assignments into Schedules}
\label{sec:transformation}

In this section, we present a transformation that turns any given assignment of jobs to machine sets into a schedule, at a bounded increase in makespan, depending on the assumptions on processing speeds, thus bounding the gap between optimal values for \textsc{Assignment} and \textsc{Scheduling} problem.
The transformation is based on the following concept of a well-structured assignment.

We say an assignment $\mathbf{S}$ is \emph{well-structured} if for every $i \in M$ there is at most one $j \in J$ with $i \in S_j$ and $|S_j| > 1$, i.e., every machine shares at most one job with other machines.
A well-structured assignment can be turned into a schedule whose makespan equals the maximum machine load: At time $0$, let every machine process its only shared job (if it exists), then process the remaining jobs assigned to each machine in arbitrary order on that machine.
Thus it suffices to show that arbitrary assignments can be transformed into well-structured assignments at a bounded increase in makespan. We first establish that this is possible for fractionally subadditive functions, in the following theorem.

\begin{theorem}
\label{thm:assignment-to-schedule-XOS}
Assume that all processing speed functions are fractionally subadditive. For any assignment $\mathbf{S}$, there exists a well-structured assignment $\mathbf{S}'$ with $L(\mathbf{S}') \leq \frac{2e}{e-1} L(\mathbf{S})$.
\end{theorem}

The proof of \cref{thm:assignment-to-schedule-XOS} is based on an interesting connection to cooperative game theory and uses the Bondareva-Shapley theorem~\citep{bondareva1963some}, which states that cost-sharing games with fractionally subadditive cost functions have a non-empty core.
In general, computing such an element of this core may require solving an exponentially sized LP.
However, in the case of submodular processing speeds, such an element can be found by a simple greedy algorithm.
With this, the entire transformation can be carried out in polynomial time for submodular processing speeds. 
We hence obtain the following theorem.

\begin{theorem}\label{thm:assignment-to-schedule}
Assume that all processing speed functions are submodular. There is an algorithm that, given an assignment $\mathbf{S}$, computes in polynomial time a well-structured assignment $\mathbf{S}'$ with $L(\mathbf{S}') \leq \frac{2e}{e-1} L(\mathbf{S})$.
\end{theorem}

Before we describe and analyze the transformation procedure that proves \cref{thm:assignment-to-schedule-XOS,thm:assignment-to-schedule}, we also remark that the transformation result for fractionally subadditive functions extends, with a logarithmic loss,  to the case that processing speeds are submodular. This immediately follows from the fact that any subadditive function can be pointwise $\ln(|M|)$-approximated by a fractionally subadditive function; see \cref{app:subadditive-transformation} for details. 

\begin{restatable}{corollary}{restatesubadditive}\label{cor:assignment-to-schedule-subadditive}
Assume that all processing speed functions are subadditive. For any assignment $\mathbf{S}$, there exists a well-structured assignment $\mathbf{S}'$ with $L(\mathbf{S}') \leq \frac{2e}{e-1}\ln(|M|) \cdot  L(\mathbf{S})$.
\end{restatable}

\subsection{Proof of \texorpdfstring{\cref{thm:assignment-to-schedule-XOS}}{Theorem \ref{thm:assignment-to-schedule-XOS}}}

Assume that $g_j$ is fractionally subadditive for all $j \in J$, let $\mathbf{S}$ be an assignment, and $C := L(\mathbf{S})$. 
The fractional subadditivity of the processing speed functions implies the following lemma.

\begin{restatable}{lemma}{restatedelta}
\label{lem:delta}
  There exists $\Delta \in \mathbb{R}^{M \times J}$ such that 
  \begin{enumerate}
    \item $\sum_{i \in S_j} \Delta_{ij} = g_j(S_j)$ for all $j \in J$,\label{eq:budget}
    \item $\sum_{i \in T} \Delta_{ij} \leq g_j(T)$ for all $j \in J$ and $T \subseteq S_j$.\label{eq:rational}
  \end{enumerate}
\end{restatable}

\begin{proof}
Noting that constraints~\ref{eq:budget} and~\ref{eq:rational} describe, for each $j \in J$, the core of a cost-sharing game with player set $S_j$ and cost function $g_j$, \cref{lem:delta} follows immediately from the Bondareva-Shapley theorem, which guarantees that the core of a  game with fractionally subadditive cost functions is non-empty. For completeness, a direct proof of the lemma is given in \cref{app:delta-existence}. 
\end{proof}

Using the coefficients $\Delta_{ij}$ given by \cref{lem:delta}, 
let $S^+_j := \{i \in S_j \st \Delta_{ij} > 0\}$ for $j \in J$ and 
consider the following feasibility LP with decision variables $x_{ij}$ for every $j \in J$ and $i \in S^+_j$.
\begin{alignat}{3}
&& \sum_{i \in S^+_j} x_{ij} & \ \geq \ 1 & \quad \forall j \in J \tag{\TLP} \label{lp:delta} \\\
&& \sum_{j \in J : i \in S^+_j} \frac{1}{\Delta_{ij}} x_{ij} & \ \leq \ C & \quad \forall\; i \in M \notag\\
&& x & \ \geq \ 0 &\notag
\end{alignat}

We first observe that \eqref{lp:delta} has a feasible solution.

\begin{lemma}
 \eqref{lp:delta} has a feasible solution.
\end{lemma}
\begin{proof}
Define $\tilde{x}$ by $\tilde{x}_{ij} := \frac{\Delta_{ij}}{g_j(S_j)}$ for $j \in J$ and $i \in S^+_j$. Note that $\sum_{i \in S^+_j} \tilde{x}_{ij} = \sum_{i \in S^+_j} \frac{\Delta_{ij}}{g_j(S_j)} = 1$ by constraint~\ref{eq:budget} of \cref{lem:delta}. Moreover, $\sum_{j \in J : i \in S^+_j} \frac{1}{\Delta_{ij}} \tilde{x}_{ij} = \sum_{j \in J : i \in S^+_j} \frac{1}{g_j(S_j)} \leq C$ by definition of $\tilde{x}$ and the fact that the maximum machine load of $\mathbf{S}$ is $C$. 
Hence $\tilde{x}$ is a feasible solution to \eqref{lp:delta}.
\end{proof}

Now consider an extreme point solution $x'$ to \eqref{lp:delta}. We will turn it into a well-structured assignment. To this end, consider the \emph{support graph} $G(x')$ of $x'$ on the node set $J \cup M$, which has an edge $\{i, j\}$ for every $j \in J$ and $i \in S^+_j$ with $x'_{ij} > 0$. 
Noting that \eqref{lp:delta} has the same structure as the assignment LP used for $R||C_{\max}$, the following lemma follows from the classic result by \citet{LST90}; see \cref{app:approximation} for details.

\begin{lemma}
  The graph $G(x')$ has an orientation such that any node has an in-degree of at most $1$.
\end{lemma}

Thus fix an orientation of $G(x')$ with in-degree at most one. 
For each job $j \in J$, let $i_j \in S_j$ denote the unique machine such that $(i_j, j)$ is an in-edge of node $j$ (if such a machine exists), and let $\bar{S_j}$ the set of machines $i \in S^+_j$ such that there is an edge $(j, i)$ in the oriented graph.
We construct a well-structured assignment $\mathbf{S}'$ as follows.
For each job $j \in J$, let $S'_j = \{i_j\}$ if $i_j$ exists and $x'_{i_j j} > 1/2$ (``Type~1 assignment'') and let $S'_j = \bar{S}_j$ otherwise (``Type~2 assignment'').
Note that, because every node has in-degree at most $1$, any machine can be assigned to at most one job via a Type~2 assignment. Hence $\mathbf{S}'$ is well-structured.

We will show that the maximum load in this assignment is bounded by $4C$ and then explain how to further modify it so as to reduce the load to the value guaranteed in the statement of \cref{thm:assignment-to-schedule-XOS}. We first argue that the load incurred by Type~1 assignments is bounded by $2C$.

\begin{lemma}\label{lem:type-1-assignment}
  For $i \in M$ let $J^1_i := \{j \in J \st i_j = i,\; x'_{i_{j}j} > 1/2\}$. Then $\sum_{j \in J^1_i} f_j(\{i\}) \leq 2C$.
\end{lemma}
\begin{proof}
The lemma follows from
$$\sum_{j \in J^1_i} f_j(\{i\}) = \sum_{j \in J^1_i} \frac{1}{g_j(\{i\})} \leq \sum_{j \in J^1_i} \frac{1}{\Delta_{ij}} \leq 2 \sum_{j \in J^1_i} \frac{1}{\Delta_{ij}} x'_{ij} \leq 2C,$$
where the first inequality follows from constraint~\ref{eq:rational} in \cref{lem:delta}, the second inequality follows from the fact that $x'_{ij} > 1/2$ for all $j \in J^1_i$, and the final inequality follows from feasibility of $x'$.
\end{proof}

Similarly, also the load incurred by Type~2 assignments can be bounded by $2C$.

\begin{lemma}\label{lem:type-2-assignment}
	Let $j \in J$ with $S'_j = \bar{S}_j$. Then $f_j(\bar{S}_j) \leq 2C$.
\end{lemma}
\begin{proof}
Note that $S'_j = \bar{S}_j$ implies that $\sum_{i \in \bar{S}_j} x'_{ij} \geq 1/2$ by construction of $\mathbf{S}'$. We conclude that
$$g_j(\bar{S}_j) \geq \sum_{i \in \bar{S}_j} \Delta_{ij} \geq \frac{1}{C} \sum_{i \in \bar{S}_j} x'_{ij} \geq \frac{1}{2C},$$
where the first inequality follows again from constraint~\ref{eq:rational} in \cref{lem:delta} and the second inequality follows from $\frac{x'_{ij}}{C} \leq \Delta_{ij}$ by feasibility of $x'$. 
\end{proof}

Because any machine participates in at most one Type~2 assignment, \cref{lem:type-1-assignment,lem:type-2-assignment} already imply that the load of any machine in the well-structured assignment $\mathbf{S}'$ is bounded by $4C$.
The maximum machine load can be further reduced by modifying $\mathbf{S}'$ as follows:
We reassign all jobs $j \in J$ that are assigned via a Type~2 assignment to $\bar{S}_j$ as follows. For each machine $i \in \bar{S}_j$, let $\ell_i := \sum_{j' \in J^1_i} \frac{1}{\Delta_{ij'}} x'_{ij'}$ be the load of machine $i$ due to jobs in $J^1_i$ in the LP solution. 
For $\theta > 0$, let $T_j(\theta) : = \{i \in \bar{S}_j \st \ell_i \leq \theta\}$.
Instead of assigning $j$ to $\bar{S}_j$, assign it to $T_j(\theta_j^*)$, where $\theta_j^*$ is a minimizer of
$\max_{i \in T_j(\theta)} 2\ell_i + f_j(T_j(\theta))$. 
\begin{restatable}{lemma}{restatefinalguarantee}
\label{lem:xos-final-guarantee}
Let $\mathbf{T}$ be the assignment defined by
$T_j = T_j(\theta^*)$ if $S'_j = \bar{S}_j$ and $T_j = \{i_j\}$ otherwise. Then $\mathbf{T}$ is well-structured and $L(\mathbf{T}) \leq \frac{2e}{e-1} L(\mathbf{S})$.
\end{restatable}
The proof of \cref{lem:xos-final-guarantee} follows the lines of the analysis for the approximation algorithm for malleable scheduling with speed-implementable processing times~\citep{fotakis2020malleable} and is presented in \cref{app:reduced-assignment}. This completes the proof of \cref{thm:assignment-to-schedule-XOS}.

\subsection{Proof of \texorpdfstring{\cref{thm:assignment-to-schedule}}{Theorem \ref{thm:assignment-to-schedule}}}

Note that the proof of \cref{thm:assignment-to-schedule-XOS} discussed above in fact describes an algorithm for turning a given assignment into a well-structured assignment:
(i) find the coefficient vector $\Delta$ as guaranteed by \cref{lem:delta}, (ii) compute an extreme point $x'$ solution to \eqref{lp:delta}, (iii) orient the support graph $G(x')$ so that every node has in-degree at most $1$, (iv) construct the assignment $\mathbf{S}'$, (v) compute optimal values $\theta^*_j$ for every $j \in J$ that is assigned via a Type~2 assignment, and (vi) reassign those jobs to the corresponding sets $T_j(\theta^*_j)$.

For steps (ii)-(iv) and (vi) it is easy to see that they can be carried out in polynomial time, even independently from any assumptions on the functions $g_j$. Also determining the values $\theta^*_j$ in step (v) can be done in polynomial time, as the sets $T_j(\theta)$ only change at $\theta \in \{\ell_i \st i \in \bar{S}_j\}$ and thus only $|\bar{S}_j| + 1$ different values have to be considered.
It remains to argue that step (i) can be carried out efficiently when $g_j$ is submodular for each $j \in J$.

Indeed, for submodular $g_j$, the constraints in \cref{lem:delta} describe a polymatroid for each $j \in J$. Thus, a feasible solution $\Delta$ can be found by the following simple greedy algorithm. Fix any arbitrary order $\preceq$ of the machines and for each $S \subseteq M$ let $\pi_i(S) := \{i' \in S \st i' \preceq i\}$.
For each $j \in J$ and $i \in S_j$, define
$\Delta_{ij} := g_j(\pi_i(S_j)) - g_j(\pi_i(S_j) \setminus \{i\})$. It is not hard to see that $\sum_{i \in S_j} \Delta_{ij} = g_j(S_j)$ for all $j \in J$ by construction. Moreover, $\sum_{i \in T} \Delta_{ij} \leq \sum_{i \in T} g_j(\pi_i(S_j \cap T)) - g_j(\pi_i(S_j \cap T) \setminus \{i\}) = g_j(T)$ for all $T \subseteq S_j$, where the inequality follows from submodularity of $g_j$. Hence, the greedily constructed coefficients $\Delta_{ij}$ fulfill the constraints of \cref{lem:delta}.
This concludes the proof of \cref{thm:assignment-to-schedule}.

\section{Approximation for Scaled Matroid-Rank Processing Speeds}
\label{sec:matroids}

In this section, we consider the case where each job $j \in J$ is associated with a matroid $\mathcal{F}_j$ on $M$ and a quota $q_j \in \mathbb{R}_{\geq 0}$. The processing speed function for job $j \in J$ is defined by
$g_j(S) := r_j(S)/q_j$,
where $r_j$ is the rank function of $\mathcal{F}_j$.

We describe an algorithm which, given an instance of the problem and a target load $C$, either computes an assignment of maximum load at most $4C$ and a schedule of makespan at most $5C$, or asserts that no assignment with maximum load at most $C$ exists. Via standard arguments (see \cref{app:approximation}), such a decision procedure implies a $4$-approximation for \textsc{Assignment} and a $5$-approximation for \textsc{Scheduling}, respectively.

\begin{theorem}\label{thm:approx-matroid}
There exists a polynomial-time $4$-approximation algorithm for \textsc{Assignment} and a $5$-approximation algorithm for \textsc{Scheduling} when all processing speed functions are scaled matroid ranks. If all jobs are furthermore identical, the approximation factor improves to $2$ for \textsc{Assignment} and $3$ for \textsc{Scheduling}, respectively.
\end{theorem}

\subsection{The Algorithm}
The algorithm splits the set of jobs in two classes: those that can be processed within time $C$ on a single machine, and those that require multiple machines to do so. It then solves the problem separately for each of the two classes, solving a classic scheduling problem for the first class and a matroid intersection problem and graph orientation problem for the second.

\subsubsection*{Partitioning the job set}
We first partition the set of jobs into two classes. To this end, define
$t_j := \left\lceil q_j/C \right\rceil \text{ for all } j \in J$
as well as $J_1 := \{j \in J \st t_j = 1\}$ and $J_2 := J \setminus J_1$.

\subsubsection*{Step 1: Single-Machine Assignments for $J_1$}
The first step of the algorithm handles the jobs in $J_1$ by constructing an instance of the classic problem of makespan minimization on unrelated machines (with non-malleable jobs).
For $i \in M$ and $j \in J_1$ let $p_{ij} = q_j$ if $\{i\} \in \mathcal{F}_j$, and $p_{ij} = \infty$ otherwise.
We set up the assignment LP for the instance of the non-malleable unrelated machine scheduling problem $R||C_{\max}$ with machine set $M$, job set $J_1$, and processing times $p$. 
\begin{align}
    \sum_{i \in M} x_{ij} \ &= \ 1 \qquad \forall j \in J_1 \tag{\LPmatroid} \label{LPmatroid}\\
    \sum_{j \in J_1} q_j x_{ij} \ &\leq \ C\qquad \forall i \in M \notag\\
    x \ &\geq \ 0 \notag
\end{align}
Notice that the fact that $t_j = 1$ implies that $q_j \leq C$ for any job $j \in J_1$. 
If the LP is feasible, we obtain an assignment $\mathbf{S}^1$ for the jobs in $J_1$ with maximum load at most $2C$ using the rounding algorithm of \citet{LST90} (see \cref{app:approximation} for details).
Otherwise, if the LP is infeasible, then we can conclude by \cref{lem:matroid-step1-feasible} (see \cref{sec:matroid-analysis}) that there is no assignment for with makespan at most $C$ for the given instance and stop the algorithm.

\subsubsection*{Step 2: Assignment via Matroid Intersection for $J_2$}
In the second step of the algorithm, we define two matroids with ground set $M \times J_2$ as follows. First, we define $\mathcal{F}'$ by
\begin{align*}
\mathcal{F}' := \big\{S \subseteq M \times J_2 \st \ & \{i \in M \st (i, j) \in S\} \,\in\, \mathcal{F}_j \ \text{ and } \\
  & |\{i \in M \st (i, j) \in S\}| \,\leq\, t_j \ \text{ for all } j \in J_2 \big\}.
\end{align*}
Note that $\mathcal{F}'$ is the direct sum of the matroids $\mathcal{F}_j$ for all $j \in J_2$, each truncated at cardinality $t_j$, respectively. In particular, $\mathcal{F}'$ is itself a matroid.
Second, we define $\mathcal{F}''$ by
\begin{align*}
  \mathcal{F}'' & := \big\{S \subseteq M \times J_2 \st \ |\{j \in J_2 \st (i, j) \in S\}| \,\leq\, 2 \ \text{ for all } i \in M\big\},
\end{align*}
i.e., $\mathcal{F}''$ is the partition matroid that allows at most two copies of each machine to be selected.
Let $B$ be a maximum cardinality common independent set in $\mathcal{F}' \cap \mathcal{F}''$ (such a set can be computed using a polynomial-time algorithm for matroid intersection~\citep{schrijver03}).
If $|B| < \sum_{j \in J_2} t_j$, then by \cref{lem:matroid-step2-feasible} below, we can conclude that there is no assignment of makespan at most $C$ for the given instance and stop the algorithm.
Otherwise, let $\mathbf{S}^2$ be the assignment for the jobs in $J_2$ induced by $B$ via $S^2_j := \{i \in M \st (i, j) \in B\}$.

\subsubsection*{Step 3a: Merging the Assignments}
If the algorithm did not terminate in either of the preceding steps with a negative answer, it combines the two assignments $\mathbf{S}^1$ and $\mathbf{S}^2$ into an assignment for all jobs in $J$ and returns this assignment. 
In the proof of \cref{thm:matroid:assignment} below, we show that the returned assignment has maximum machine load of at most $4C$.
This assignment could be turned into a schedule of makespan at most $8e/(e-1) \cdot C$ via \cref{thm:assignment-to-schedule}. However, the more careful procedure described below shows that we can obtain a (well-structured) schedule of makespan at most $5C$ by a slightly different way.

\subsubsection*{Step 3b: Creating a Schedule}
For each $i \in M$ let $e_i := \{j \in J_2 \st (i, j) \in B\}$, where $B$ is the set constructed in Step~2. Note that $|e_i| \leq 2$ for each $i \in M$ as $B \in \mathcal{F}''$, hence the set $E := \{e_i \st i \in M,\, e_i \neq \emptyset \}$ can be interpreted as a set of edges on the node set $J_2$, where we interpret possible singletons $e_i = \{j\}$ as self-loops of the node $j$. Let~$G = (J_2, E)$ be the corresponding undirected graph. Note that each node $j \in J_2$ has a degree of $t_j$ (counting self-loops only once).
Hence, by a classic result of \citeauthor{hakimi65}~(\citeyear{hakimi65}; also see \cref{app:graph-orientation}), we can compute an orientation $D$ of $G$ in which every node has in-degree at least $\lfloor t_j / 2 \rfloor$. Let $j_i \in J_2$ be the head node of the directed arc in the orientation $D$ that corresponds to undirected edge $e_i \in E$. Let $\bar{B} := \{(i, j_i) \st i \in M,\, e_i \neq \emptyset\}$. Note that $\bar{B} \subseteq B$ by definition, and that $|\{j : (i, j) \in \bar{B}\}| \leq 1$ for all $i \in M$ and $|\{i : (i, j) \in \bar{B}\}| \geq \lfloor t_j / 2 \rfloor$ for all $j \in J_2$.
Hence $\bar{\mathbf{S}}^2$ defined by $\bar{S}^2_j := \{i \in M \st (i, j) \in \bar{B}\}$ is a well-structured assignment on $J_2$. Because every job in $J_1$ is assigned to a single machine in $\mathbf{S}^1$, combining $\mathbf{S}^1$ and $\mathbf{S}^2$ yields a well-structured assignment that can be turned into a schedule. 
\cref{thm:matroid:schedule} argues that the resulting schedule has makespan at most $5$.

Finally, we remark that for identical jobs, either $J_1$ or $J_2$ is empty, and therefore the approximation factor improves to $2$ for \textsc{Assignment} and $3$ for \textsc{Scheduling}, respectively.

\subsection{Analysis}
\label{sec:matroid-analysis}

In the remainder of this section, we complete the analysis of the algorithm described above. We start by showing that \eqref{LPmatroid} constructed in Step~1 of the algorithm is indeed feasible if there is an assignment of maximum load $C$, and hence the algorithm correctly terminates if the LP is infeasible.

\begin{lemma}\label{lem:matroid-step1-feasible}
If there is an assignment of maximum load at most $C$, then \eqref{LPmatroid} has a feasible solution.
\end{lemma}
\begin{proof}
Let $\mathbf{S}$ be an an assignment with maximum load at most $C$.
We first observe that, without loss of generality, $S_j \in \mathcal{F}_j$ for each $j \in J$ (otherwise some machines can be removed from $S_j$ without increasing the processing time).
We construct a feasible solution to \eqref{LPmatroid} as follows.
For any $i \in M$ and $j \in J$, we set $x'_{ij} = \frac{1}{|S_j|}$ if $i \in S_j$, and $x'_{ij} = 0$ otherwise. For this assignment, the first and third set of constraints of \eqref{LPmatroid} are trivially satisfied. Finally, for the second set of constraints we obtain:
\begin{align*}
\sum_{j \in J_1} q_j x'_{ij} = \sum_{j \in J_1 : i \in S_j} \frac{q_j}{|S_j|} = \sum_{j \in J_1 : i \in S_j} \frac{q_j}{r_j(S_j)} = \sum_{j \in J_1 : i \in S_j} f_j(S_j) \leq C,
\end{align*}
for every $i \in M$, where the second equality follows from our assumption that $S_j \in \mathcal{F}_j$ and the last inequality by the feasibility of the optimal solution.
\end{proof}

The next lemma shows that if the algorithm terminates in Step~2 it correctly concludes that there is no assignment with maximum load $C$.

\begin{lemma}\label{lem:matroid-step2-feasible}
If there is an assignment of makespan at most $C$, then there is a $B \in \mathcal{F}' \cap \mathcal{F}''$ with $|B| = \sum_{j \in J_2} t_j$.
\end{lemma}
\begin{proof}
Let $\mathbf{S}$ be an an assignment with maximum load at most $C$. 
As in the proof of \cref{lem:matroid-step1-feasible}, we assume without loss of generality that $S_j \in \mathcal{F}_j$. 

For $j \in J_2$ and $i \in S_j$, we define $x'_{ij} = t_j / |S_j|$. Consider the following LP:
\begin{alignat*}{3}
&& \sum_{i \in S_j} x_{ij} & \ = \ t_j & \quad \forall\, j \in J_2 \\
&& \sum_{j \in J_2 : i \in S_j} x_{ij} & \ \leq \ 2 & \quad \forall\, i \in M \\
&& 0 \ \leq & \ x_{ij} \ \leq \ 1 & \quad \forall\, j \in J_2,\, i \in S_j
\end{alignat*}
Note that $x'$ is a feasible solution to the LP. The first type of constraints is trivially satisfied by $x'$, while for the third set of constraints we have $0 \leq x'_{ij} = \frac{t_j}{|S_j|} = \frac{t_j}{r(S_j)} \leq 1$ for each $j \in J_2$ and $i \in M$, since $r(S_j) \geq \lceil q_j / C \rceil = t_j$ in any feasible solution. 

Furthermore, note that
$$f_j(S_j) = \frac{q_j}{|S_j|} \geq \frac{(t_j - 1) C}{|S_j|} \geq \frac{C}{2} \cdot \frac{t_j}{|S_j|} = \frac{C}{2} x'_{ij}$$
for all $j \in J_2$ and $i \in S_j$ and, hence,
$\sum_{j \in J_2} x'_{ij} \leq \frac{2}{C} \sum_{j \in J_2} f_j(S_j) \leq \frac{2}{C} \cdot C$
by feasibility of the assignment for makespan $C$.

Notice that the constraint matrix of the above LP is totally unimodular, as it corresponds to the adjacency matrix of a bipartite graph. Thus, since the right hand side of the constraints is integral, we know that the LP has an integral solution $x^*$. Let $S^*_j = \{i \in M \st x^*_{ij} = 1\}$. Note that $S^*_j \subseteq S_j \in \mathcal{F}_j$ and $|S^*_j| = t_j$ for all $j \in J_2$.
Hence $B = \{(i, j) \st j \in J_2,\, i \in S^*_j\} \in \mathcal{F}'$.
Furthermore, $|\{j \in J_2 \st i \in S^*_j\}| \leq 2$ and, thus, $B \in \mathcal{F}''$. 
\end{proof}

We now show that the assignment constructed in Step~3a has maximum load of at most $4C$, establishing the first part of \cref{thm:approx-matroid}.

\begin{theorem}\label{thm:matroid:assignment}
If the algorithm reaches Step~3a, it constructs an assignment of maximum machine load at most $4C$.
\end{theorem}
\begin{proof}
Let $\mathbf{S}$ be the assignment for all jobs in $J$ arising from combining $\mathbf{S}^1$ and $\mathbf{S}^2$.
Note that $S_j \in \mathcal{F}_j$ and $|S_j| \leq t_j$ for all $j \in J_2$ by definition of $\mathcal{F}'$ and construction of $\mathbf{S}^2$ from $B \in \mathcal{F}'$.
Thus, $|B| = \sum_{j \in J_2} t_j$ implies that $|S_j| = t_j$ for all $j \in J_2$.
Therefore $f_j(S_j) = \frac{q_j}{r_j(S_j)} = \frac{q_j}{t_j} \leq C$ for each $j \in J_2$.
Moreover, because $|\{j \in J_2 \st i \in S_j\}| \leq 2$ for each $i \in M$ by definition of $\mathcal{F}''$ and construction of $\mathbf{S}^2$ from $B \in \mathcal{F}''$.
We conclude that $\sum_{j \in J_2 : i \in S_j} f_j(S_j) \leq 2C$ for each $i \in M$.
Because the maximum machine load of the assignment $\mathbf{S}^1$ is also bounded by $2C$, we obtain $\sum_{j \in J : i \in S_j} f_j(S_j) \leq 4C$. 
\end{proof}

Finally, we show that the well-structured assignment constructed in Step~3b has maximum load of at most $5C$, establishing the second part of \cref{thm:approx-matroid}.

\begin{theorem} \label{thm:matroid:schedule}
If the algorithm reaches Step~3b, it constructs a well-structured assignment with maximum machine load at most $5C$.
\end{theorem}
\begin{proof}
As already argued in the description of Step~3b, the produced assignment is well-structured. 
Because the load of each machine in the assignment $\mathbf{S}_1$ is bounded by $2C$, it remains to show that $\sum_{j \in J_2 : i \in \bar{S}^2_j} f_j(\bar{S}^2_j) \leq 3C$ for each $i \in M$.
Note that for $i \in M$ there is at most one $j \in J_2$ with $i \in \bar{S}^2_j$.
Moreover,
$f_j(S_j) = q_j / r_j(S_j) \leq q_j / \lfloor t_j/2 \rfloor \leq 3 q_j / t_j \leq 3 C$
for each $j \in J_2$,
where we use the definition of $t_j \geq 2$ for $j \in J_2$ and the fact that $\lfloor t_j / 2 \rfloor \geq t_j / 3$ for any non-negative integer $t_j \geq 2$. Hence, we can conclude that the maximum machine load is bounded by $5 C$. 
\end{proof}

\section{Approximation for Subadditive or Submodular Speeds}
\label{sec:submodular}

In this section, we develop a practically efficient heuristic for both \textsc{Assignment} and \textsc{Scheduling} for a class of processing speed functions that includes submodular functions and monotone subadditive functions.
The algorithm yields logarithmic approximation guarantees for submodular or monotone fractionally subadditive processing speeds, and polylogarithmic approximation guarantees for monotone subadditive processing speeds. 
Furthermore, a simple modification of our heuristic yields improved approximation guarantees for the case of identical jobs.

For fractionally subadditive or subadditive processing speeds, the algorithm requires the access to demand oracle for each speed function $g$, which given a vector $p \in \mathbb{R}^M$ returns a maximizer of $g(S) - \sum_{i \in S} p_i$. 
We point out that access to such a stronger oracle is necessary for (fractionally) subadditive processing speeds, as for these functions, no meaningful approximation guarantees can be achieved using the standard value oracle within a polynomial number of oracle queries, as we show in the following result (see \cref{app:hardness} for the proof): 

\begin{restatable}{theorem}{restatehardness}
\label{thm:hardness}
Any algorithm with approximation guarantee $o(n^{\frac{1}{3}})$ for monotone fractionally subadditive processing speeds requires an exponential number of queries to a value oracle. 
\end{restatable}

Before we present our algorithm and its analysis, we first define the class of instances to which our results apply.

\paragraph{$(\alpha, \beta)$-good functions.}\label{def:alphabetainstance} Let $\mathcal{G}$ be a class of set functions. 
We say that $\mathcal{G}$ is \emph{$(\alpha, \beta)$-good} for $\alpha, \beta \geq 1$ if the following three properties are fulfilled:
\begin{enumerate}
	\item For any instance of \textsc{Assignment} with processing speeds from $\mathcal{G}$ and any assignment $\mathbf{S}$ for that instance, there exists well-structured assignment $\mathbf{S}'$ with $L(\mathbf{S}') \leq \alpha L(\mathbf{S})$. 
	\item The class of functions $\mathcal{G}$ is \emph{closed under truncation}, i.e., for any $g \in \mathcal{G}$ and any $t \in \mathbb{R}_{\geq 0}$, the function $g^t$ defined by $g^t(S) := \min \{g(S), t\}$ is contained in $\mathcal{G}$.
	\item There exists a $\beta$-approximation algorithm for the following \textsc{Welfare Maximization} problem:
	Given a set of set functions $g_1, \dots, g_k$ from $\mathcal{G}$ on the same ground set $M$, find $k$ disjoint (possibly empty) subsets $T_j \subseteq M$ for $j \in \{1, \dots, k\}$ so as to maximize $\sum_{j = 1}^{k} g_j(T_j)$.
\end{enumerate}

Our algorithm and its analysis will establish the following metatheorem, which implies the approximation results for subadditive, fractionally subadditive, and submodular processing speeds, resepectively, as we show further below.
For the remainder of this section, let $n := |J|$ and $m := |M|$.

\begin{theorem}\label{thm:meta-heuristic}
Let $\mathcal{G}$ be an $(\alpha, \beta)$-good class of set functions. Then there exists a polynomial-time $\mathcal{O}(\alpha \beta^2 \log \min\{n, m\})$-approximation algorithm for \textsc{Assignment} and \textsc{Scheduling} restricted to instances with processing speeds from $\mathcal{G}$.
\end{theorem}

At the heart of the algorithm lies a subroutine that concurrently runs two assignment procedures: one that tries to assign a maximum number of jobs to individual machines within a given makespan and another that tries to assign a large number of jobs to disjoint sets of machines via a \textsc{Welfare Maximization} problem. These procedures are described and analyzed in \cref{sec:submodular-phase-proof}, where we show that properties of $(\alpha, \beta)$-good functions ensure that a significant fraction of the jobs can be assigned by at least one of the two procedures.
This establishes the following lemma, whose repeated application implies \cref{thm:meta-heuristic}.

\begin{lemma}
\label{lem:submodular-phase}
Let $\mathcal{G}$ be an $(\alpha, \beta)$-good class of set functions.
There is a polynomial-time algorithm that, given an instance of \textsc{Assignment} with processing speeds form $\mathcal{G}$ and $C > 0$, returns a subset $J' \subseteq J$ with $|J'| \geq \max\{n - m, \frac{n}{4\beta}\}$ and a well-structured assignment $\mathbf{S}'$ for the jobs in $J'$ with $L(\mathbf{S}') \leq 4\alpha \beta C$, or correctly outputs that no assignment with maximum load at most $C$ exists for the jobs in $J$.
\end{lemma}

The algorithm proving \cref{thm:meta-heuristic} runs in phases.
In each phase, the algorithm established in the proof of \cref{lem:submodular-phase} is applied to assign a subset of the still unassigned jobs (initially $J$). \cref{lem:submodular-phase} guarantees that in each phase, a $\frac{1}{4\beta}$-fraction of the remaining jobs is assigned.
Therefore, after at most $\log_{4\beta / (4\beta - 1)} n = \mathcal{O}(\beta \log n)$ phases, all jobs are assigned. Moreover, at most $m$ jobs remain after the first phase and thus the number of phases is also bounded by $\mathcal{O}(\beta \log m)$.
As the assignments constructed in each phase have a maximum load of at most $4\alpha \beta C$, the concatenation of all these assignments results in an assignment of maximum load at most $\mathcal{O}(\alpha \beta^2 \log \min\{n, m\}C)$. As the partial assignments constructed in each phase are also well-structured, they can be turned into schedules and concatenated without increasing the makespan; see \cref{app:approximation} for details.

\subsection{Consequences of \texorpdfstring{\cref{thm:meta-heuristic}}{Theorem \ref{thm:meta-heuristic}}}

We now present several corollaries of \cref{thm:meta-heuristic} for the various classes of speed functions of interest. We remark that the classes of subadditive, fractionally subadditive, and submodular set functions are all closed under truncation. 

We first observe that submodular functions given by a value oracle are $(\frac{2e}{e-1}, 4+\varepsilon)$-good for any $\varepsilon > 0$. This follows from \cref{thm:assignment-to-schedule-XOS} and the $(4+\varepsilon)$-approximation algorithm by \citet{LMNS09} for \textsc{Submodular Maximization Over Matroid Constraints}, which generalizes the \textsc{Welfare Maximization} problem for (possibly non-monotone) submodular functions.

\begin{corollary}
There exists a polynomial-time $\mathcal{O}(\log \min\{n,m\})$-approximation algorithm for \textsc{Assignment} and \textsc{Scheduling} restricted to submodular processing speeds.
\end{corollary}

For monotone fractionally subadditive functions given by a demand oracle, we can obtain similar results. This follows from \cref{thm:assignment-to-schedule-XOS}, which allows us to set $\alpha = \frac{2e}{e-1}$, and the $2$-approximation for \textsc{Welfare Maximization} for monotone subadditive functions using a demand oracle, which allows us to set $\beta = 2$.

\begin{corollary}
There exists a polynomial-time $\mathcal{O}(\log \min\{n,m\})$-approximation algorithm for \textsc{Assignment} and \textsc{Scheduling} restricted to monotone fractionally subadditive processing speeds given by a demand oracle.
\end{corollary}

Finally, note that monotone subadditive functions given via a demand oracle are $(\frac{2e}{e-1}\ln(m), 2)$-good. This follows from \cref{cor:assignment-to-schedule-subadditive} and $2$-approximation for \textsc{Welfare Maximization} by \citet{feige2009maximizing}, mentioned above. 
\begin{corollary}
There exists a polynomial-time $\mathcal{O}(\log m \log \min\{n,m\})$-approximation algorithm for \textsc{Assignment} and \textsc{Scheduling} restricted to monotone subadditive processing speeds given by a demand oracle.
\end{corollary}

For the case of identical jobs (i.e., with identical processing speed functions) a simple modification of the above heuristic yields a $\mathcal{O}(1)$-approximation for \textsc{Assignment} and \textsc{Scheduling}. Indeed, notice that, by \cref{lem:submodular-phase}, at least $\frac{n}{4\beta}$ jobs are scheduled with a makespan at most $4 \alpha \beta C$ during the first phase of the heuristic. Now, instead of moving on to the subsequent phase and repeating the process for the remaining jobs, we create $4\beta$ identical copies of the schedule produced in the first phase, and return their concatenation as the final schedule. Since the jobs are identical (and, thus, interchangeable), $4\beta$ copies suffice for scheduling all of them with makespan at most $16 \alpha \beta^2 C$. This observation, in combination with the above discussion on subadditive and submodular functions, immediately yields the following results.

\begin{corollary}
There exists a polynomial-time $\mathcal{O}(1)$-approximation algorithm for \textsc{Assignment} when all jobs are identical with submodular processing speeds, or when all jobs are identical with monotone fractionally subadditive processing speeds that are given by demand oracle.
\end{corollary}

\begin{corollary}
There exists a polynomial-time $\mathcal{O}(\log m)$-approximation algorithm for \textsc{Assignment} and \textsc{Scheduling} when all jobs are identical with monotone subadditive processing speed function that are given by a demand oracle.
\end{corollary}

\subsection{Proof of \texorpdfstring{\cref{lem:submodular-phase}}{Lemma \ref{lem:submodular-phase}}}
\label{sec:submodular-phase-proof}

The lemma is based on the following corollary, which is a direct consequence of the properties of $(\alpha, \beta)$-good functions, which guarantee the existence of an $\alpha$-approximate well-structured assignment.
\begin{corollary}\label{cor:well-structured-partition}
Assume processing speeds are $(\alpha, \beta)$-good and there is an assignment of maximum load $C$.
Then there is an assignment $\mathbf{S}$ and sets $J_1, J_2 \subseteq J$ with $J = J_1 \cup J_2$ such that
  \begin{itemize}
    \item $|S_j| = 1$ for all $j \in J_1$ and $\sum_{j \in J_1 : j \in S_j} f_j(S_j) \leq \alpha C$ for all $i \in M$,
    \item $S_j \cap S_{j'} = \emptyset$ for all $j, j' \in J_2$ with $j \neq j'$ and $f_j(S_j) \leq \alpha C$ for all $j \in J_2$.
  \end{itemize}
\end{corollary}
Thus, there exists an assignment such that at least half the jobs are assigned to individual machines with maximum load $\alpha C$, or at least half the jobs are assigned to disjoint sets of machines, each with a processing time of at most $\alpha C$. 
As a consequence, one of the two assignment procedures described below will be able to assign either at least $\frac{n}{2}$ jobs with maximum load $2\alpha C$, or at least $\frac{n}{4\beta}$ jobs with maximum load of $4\alpha \beta C$.

The first assignment procedure, which assigns a maximal number of jobs to individual machines within a given target makespan, is based on the following result that follows directly from a classic approximation result for the \textsc{Generalized Assignment} problem~\citep{shmoys1993approximation}; see \cref{app:submodular-gap} for details.

\begin{restatable}{lemma}{restatenosplitting}\label{lem:submodular:nosplitting} Assume there exists an assignment $\mathbf{S}$ for some set of jobs $J_1 \subseteq J$ with maximum load $C'$ such that $|S_j| = 1$ for all $j \in J_1$. Then we can construct in polynomial-time an assignment that assigns at least $|J_1|$ jobs (possibly different from $J_1$) with a maximum load at most $2 C'$.
\end{restatable}

The second assignment procedure, which assigns jobs to disjoint sets of machines, is based on an approximation algorithm for the \textsc{Welfare Maximization} problem for submodular funtions.

\begin{lemma}\label{lem:submodular:splitting}
For any $(\alpha, \beta)$-instance of our problem, assume there exists a set $J_2 \subseteq J$ with $|J_2| \geq |J|/2$ and an assignment $\mathbf{S}$ of the jobs in $J_2$ with maximum load $C'$ such that $|\{j \in J_2 \st i \in S_j\}| \leq 1$ for each $i \in M$. 
Then we can construct in polynomial time a set of jobs $J' \subseteq J$ with $|J'| \geq |J|/(4 \beta - 1)$ and an assignment $\mathbf{S}'$ of the jobs in $J'$ with maximum load at most $4 \beta C'$.
\end{lemma}
\begin{proof}
For $j \in J$ and $S \subseteq M$ define $h_j(S) := \min\{g_j(S) , 1/C'\}$, namely, the truncation of $g_{j}(S)$ at $1/C'$. Consider the \textsc{Welfare Maximization} problem, which asks for disjoint (possibly empty) subsets $T_j \subseteq M$ for $j \in J$ so as to maximize $\sum_{j \in J} h_j(T_j)$.
Note that setting $T_j = S_j$ for $j \in J_2$ and $T_j = \emptyset$ for $j \in J \setminus J_2$ induces a feasible solution of value $\sum_{j \in J'} h_j(S_j) \geq |J|/(2C')$.

If there exists a $\beta$-approximation for the \textsc{Welfare Maximization} problem for the class of functions $\{g_j\}_{j \in J}$, then the same approximation holds for $\{h_j\}_{j \in J}$, since the former class is closed under truncation, by \cref{def:alphabetainstance}. Thus, we can apply the $\beta$-approximation for the \textsc{Welfare Maximization} problem to obtain disjoint subsets $T'_j$ for $j \in J$ with $\sum_{j \in J} h_j(T'_j) \geq |J_2|/C' \geq |J| /(2 \beta C')$.

Let $J' := \{j \in J \st h_j(T'_j) \geq 1/(4 \beta C')\}$ and note that
$$\frac{|J|}{2 \beta C'} \leq \sum_{j \in J} h_j(T'_j) \leq \frac{|J'|}{C'} + \frac{|J \setminus J'|}{4 \beta C'},$$
where the second inequality follows from the fact that $h_j(S) \leq 1/C'$ for all $S \subseteq M$.
Rearranging yields $|J| \leq (4 \beta - 1) |J'|$, hence setting $S'_j = T'_j$ for $j \in J'$ yields the statement of the lemma.
\end{proof}

By \cref{cor:well-structured-partition}, either the premise of \cref{lem:submodular:splitting} is fulfilled with $C' = \alpha C$, implying $|J|/(4\beta - 1)$ jobs can be assigned within makespan $(4\alpha\beta) C$, or \cref{lem:submodular:nosplitting}, again with $C' = \alpha C$, guarantees the assignment of $|J_1| = |J \setminus J_2| \geq |J|/2$ jobs in time $2 \alpha C$. Finally, for sets $J_1$ and $J_2$ of \cref{cor:well-structured-partition} it holds $|J_2| \leq m$ (as jobs in $J_2$ are processed on disjoint machine sets) and hence, $|J_1| = |J| - |J_2| \geq |J| - m$, where $m$ is the number of machines. Therefore, $J'$ must also satisfy $|J'| = \max\{|J_1|, |J_2|\} \geq |J| - m$. This concludes the proof of \cref{lem:submodular-phase}.

\section{Computational Study}
\label{sec:computational}

In this section, we provide an empirical evaluation of our $\mathcal{O}(\log \min\{ n, m \})$-approximation algorithm presented in \cref{sec:submodular} on synthetic data. We start by describing the submodular speed functions we consider and the LP-based lower bound we use as a baseline. Using this baseline, we then present the empirical approximation ratios for the \textsc{Assignment} and \textsc{Scheduling} problem for each of these functions on randomly generated input. Finally, we provide examples of the application of \cref{thm:assignment-to-schedule} on small instances where the optimal assignment can be computed via an MILP formulation.

\subsection{Implementation Details}
We remark that we use a simple variant of our $\mathcal{O}(\log \min\{ n, m \})$-approximation heuristic, where we use the standard greedy $\frac{1}{2}$-approximation algorithm for solving the submodular welfare problem (with monotone functions) of \cref{lem:submodular:splitting}. Even though using the $\left(1 - \frac{1}{e}\right)$-approximation algorithm due to \cite{Von08} might lead to better constants in the approximation guarantee, the greedy approach is significantly more practical and easy to implement. For solving the scheduling problem, we first compute an approximate assignment using the heuristic of \cref{sec:submodular} and, then, we apply the transformation of \cref{sec:transformation} on the produced assignment.

The algorithm has been implemented in Python 3.8. We have used Gurobi$^{\text{TM}}$ 9.1 Solver for all the associated optimization problems (i.e., computation of the optimal solution via an MILP formulation, computation of LP-based lower bounds, and implementation of the algorithm). Finally, we have used SageMath for implementing various operations on matroids. The experiments have been conducted on a machine with 3.2 GHz 8-core CPU and 8 GB RAM.

\subsection{Instance Sets}

We compute the empirical approximation ratio of our heuristic for a varying number of $n$ jobs and $m$ machines on the following three classes of processing speed functions. The functions from each of these three classes are known to be monotone submodular.

\begin{itemize}
    \item \textsc{Coverage}: We consider $k$ {\em slots} and a {\em frequency} parameter $p \in (0,1)$. For each job $j$ and machine $i$, let $B_{i,j} \subseteq [k]$ be a subset of the $k$ slots, such that each slot is contained in $B_{i,j}$ independently with probability $p$. The speed function of each job $j$ on a set $S \subseteq M$ is defined as
    \begin{align*}
        g_j(S) = \frac{1}{L_j} \cdot \bigg| \bigcup_{i \in M} B_{i,j} \bigg|,
    \end{align*}
    where $L_j$ is a random constant indicating the {\em total load} of job $j$.
    
    In the following simulations, we choose frequency $p = 0.2$ and number of slots $k = m$, where $m$ is the number of machines. For each job $j$, we choose load $L_j \sim U[1,1000]$.

    \item \textsc{Budget-Additive}: For each job $j$, we define a {\em budget} $W_j$ and a set of {\em weights} $w_{i,j}$ for each $i \in M$. For the budget and the weights, we have $W_j \sim U[1, W_{\max}]$ and $w_{i,j} \sim U[1, w_{\max}]$, where $W_{\max}, w_{\max}$ fixed constants, with $W_{\max} \gg w_{\max}$. The speed function of each job $j$ of load $L_j$ on a set $S \subseteq M$ is defined as
    \begin{align*}
        g_j(S) = \frac{1}{L_j} \cdot \min\bigg\{\sum_{i \in S} w_{i,j}, W_j\bigg\}.
    \end{align*}
    
    Specifically, for each job $j$, we choose $L_j \sim U[1,1000]$, $w_{\max} = 100$, and $W_{\max} = 100 \cdot m$, where $m$ is the number of machines.

    \item \textsc{Matroid-Matching}: For each job $j$, we define $k$ {\em slots}, each associated with a weight $w_{\ell,j} \sim U[1, w_{\max}]$. We define a bipartite graph $G_j$ between the set of slots, $[k]$, and the set of machines, $M$. For each slot $\ell$ and machine $i$, with probability $p$, we consider an edge $\{\ell,i\}$ with weight equal to $w_{\ell,j}$ (such that all the edges adjacent to the same slot have the same weight). For every matching $\sigma \subseteq [k] \times M$ on $G_j$, we denote by $w_j(\sigma)$ the total weight of the edges is $\sigma$. 
    
    Further, we consider a random partition $\bigcup_{\ell \in [b]} B_{\ell,j}$ on the $k$ slots, such that each slot is placed into some group uniformly at random. We define a {\em partition matroid} on the slots such that any set $R \subseteq [k]$ is {\em independent} if $|R \cap B_{\ell,j}| \leq r$ for each $\ell \in [b]$, for some {\em rank} $r$. Let $\mathcal{I}_j$ be the family of independent sets in the above matroid.
    
For some random load $L_j$, the speed function for job $j$ is defined as 
\begin{align*}
    g_j(S) = \frac{1}{L_j} \cdot \max\bigg\{w_j(\sigma) \mid \sigma \subseteq [k] \times M \text{ a matching on $G_j$ between }R \in \mathcal{I}_j \text{ and } S'\subseteq S \bigg\} .
\end{align*}

In our simulations, we set $w_{\max} = 100$, $p = 0.1$, $r = 2$, $k = \lfloor \frac{m}{4} \rfloor + 1$, and $b = \lfloor \frac{m}{8} \rfloor + 1$, where $m$ is the number of machines. Finally, for each job $j$, we have $L_j \sim U[1,1000]$.

\end{itemize}

\subsection{Lower Bounds}

In order to estimate the empirical approximation ratio of our heuristic, we need to be able to compute (or lower bound) the makespan of the optimal assignment in any given instance. 

We take two different routes depending on the size of the instance:
\begin{itemize}
    \item \textbf{MILP formulation}: For relatively small instances, we are able to exactly compute the makespan of the optimal assignment by constructing the following {\em mixed integer linear programming} (MILP) formulation: 
    \begin{alignat*}{3}
        \minimize \ &&  C & &\\
        \subjectto \ && \sum_{S \subseteq M} x(S, j) & \ \geq \ 1 & \quad \forall j \in J \\
        && \sum_{j \in J} \sum_{S \subseteq M : i \in M} \frac{1}{g_j(S)} x(S, j) & \ \leq \ C & \quad \forall i \in M \\
        && x(S,j) & \ \in \{0,1\} & \quad \forall j \in J,~ S \subseteq M \setminus \{\emptyset\}.
    \end{alignat*}
    
    Note that in order to model all possible assignments from jobs to (subsets of) machines, we need exactly $n \cdot (2^m-1)$ binary variables. Unfortunately, the exponential dependence on $m$ allows the usage of the above formulation only for a small number of machines.  
    
    \item \textbf{LP relaxation}: For larger instances where the MILP approach becomes computationally intractable, we can compute a lower bound to the makespan via the following relaxation:
    
    \begin{alignat*}{3}
        \minimize \ &&  C & &\\
        \subjectto \ && \sum_{i \in M} x_{ij} & \ \geq \ 1 & \quad \forall j \in J \\
        && \sum_{j \in J} \frac{1}{g_j(\{i\})} x_{ij} & \ \leq \ C & \quad \forall i \in M \\
        && \max_{j \in J} \frac{1}{g_j(M)} & \ \leq \ C & \quad \forall j \in J \\
        && x & \ \geq 0 &.
    \end{alignat*}

Let $C^*$ be the optimal solution of the above LP. It is not hard to see that $C^*$ is essentially the maximum between a naive relaxation of the problem, where each job $j$ can be fractionally assigned to each machine $i$ with processing time ${1}/{g_j(\{i\})}$, and the maximum processing time of a job using the whole set of machines.

We remark that $C^*$ can be quite far from the optimal solution, particularly in the case where the number of jobs is smaller than the number of machines -- a fact which is reflected in our simulations. Nevertheless, the above lower bound appears to be sufficient for showing that the empirical approximation ratio of our algorithms relatively small in practice.
\end{itemize}

\subsection{Empirical Evaluation of the \texorpdfstring{$\mathcal{O}(\log \min\{n,m\})$}{O(log min\{n,m\})}-Approximation Heuristic}

For each of the instances described above we computed a feasible assignment using the heuristic from \cref{sec:submodular} and compared this against the LP-based lower bound described above. The observed ratios of assignment makespan to the respective lower bound for that instance, grouped by number of jobs 
and machines, 
are plotted in subfigure (a) of \cref{fig:plotsCoverage,fig:plotsBudgeted,fig:plotsMatching} for the three different types of speed functions, respectively.
We further transformed each of the resulting assignments into a feasible schedule using the algorithm described in \cref{thm:assignment-to-schedule}. Again we compare the makespan of the resulting schedule against the LP-based lower bound on the makespan of an optimal assignment. The observed ratios, again grouped by number of jobs and machines, are plotted in subfigure (b) of \cref{fig:plotsCoverage,fig:plotsBudgeted,fig:plotsMatching} for the different types of speed functions, respectively.

As we observe, the empirical approximation ratios reported in \cref{fig:plotsCoverage,fig:plotsBudgeted,fig:plotsMatching} take significantly smaller values comparing to the theoretically proven guarantees. It is important to note that, for the scheduling problem, we use as a baseline the same LP-based lower bound used for the assignment problem. This fact implies that the actual approximation ratios for this case can be even smaller than they appear. Finally, for all the speed functions, and especially for \textsc{Coverage}, we observe the following interesting phenomenon: the empirical approximation guarantee appears to be higher when the number of jobs is smaller than the number of machines. In particular, we believe that the above phenomenon is partially explained by the fact that the LP-based lower bound we use is weak for these instances. Indeed, this is because the lower bound interpolates between the case of scheduling the heaviest job on all the machines (which is weak for more than one jobs), and the case where each job can be fractionally assigned to more than one machines (which has an unbounded integrality gap). 

\begin{figure}
\centering
    \includegraphics[width=\textwidth]{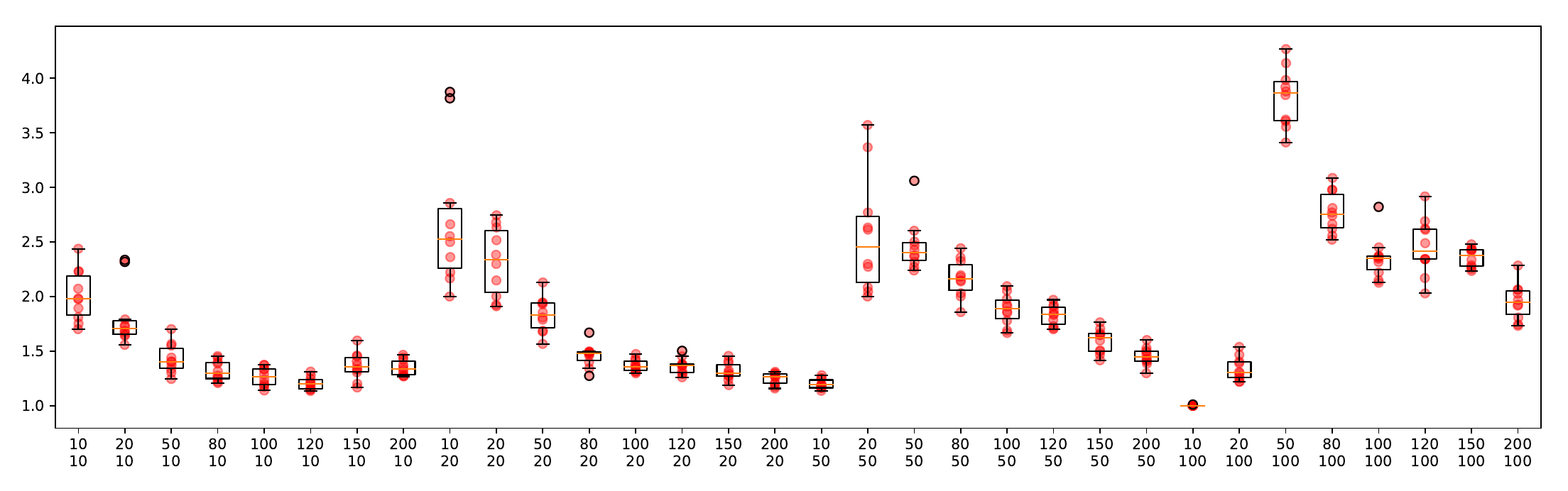}
    
    (a) Assignment using the heuristic of \cref{sec:submodular}.
    \includegraphics[width=\textwidth]{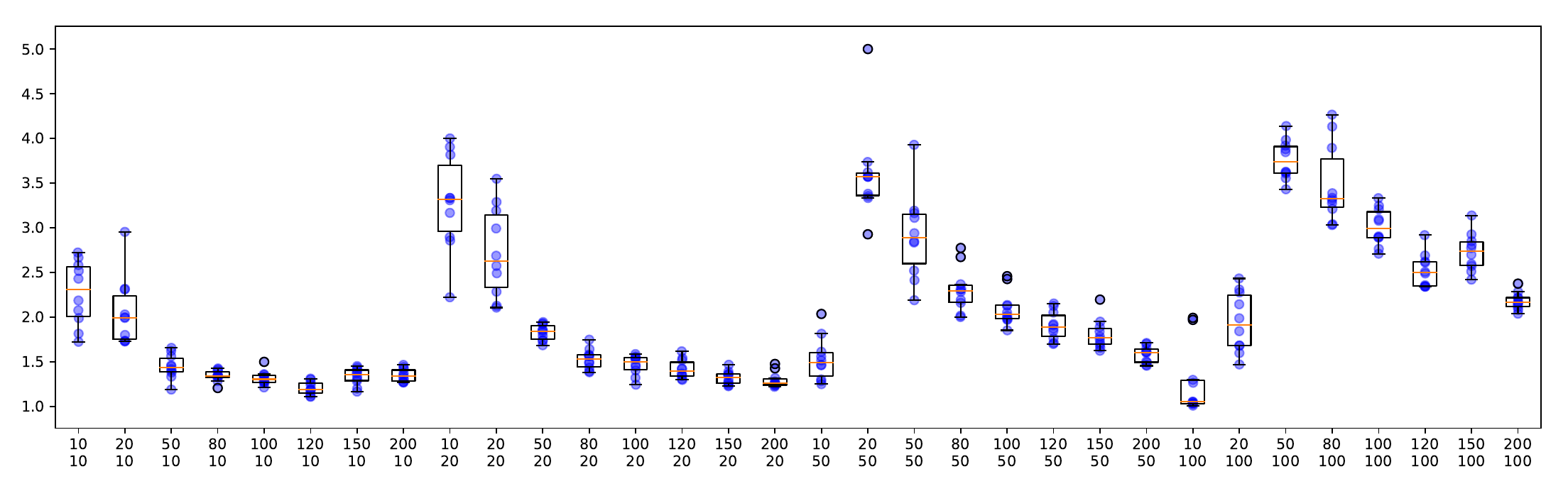}
    
    (b) Scheduling via applying \cref{thm:assignment-to-schedule} to the produced assignments.
    \caption{Empirical approximation ratio for \textsc{Coverage} speed functions. Each point in the $x$-axis is labeled as $\genfrac{}{}{0pt}{}{n}{m}$, where $n$ and $m$ is the number of jobs and machines, respectively.
    }
    \label{fig:plotsCoverage}
\end{figure}
\begin{figure}[t]
\centering
    \includegraphics[width=\textwidth]{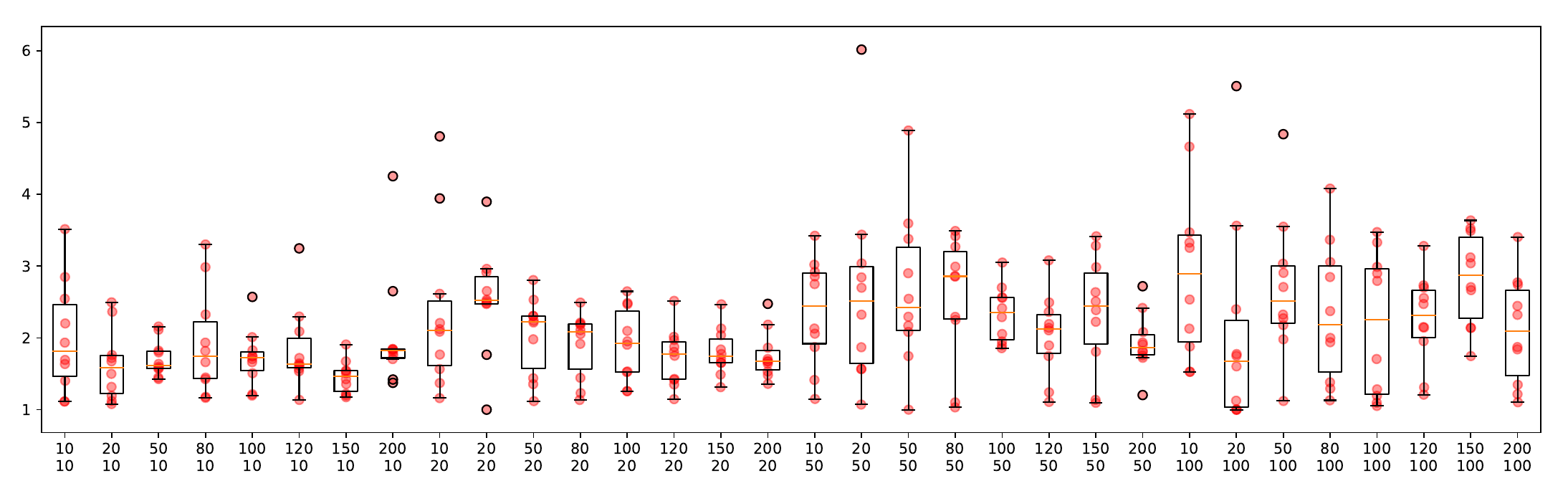}
    
    (a) Assignment using the heuristic of \cref{sec:submodular}.
    \includegraphics[width=\textwidth]{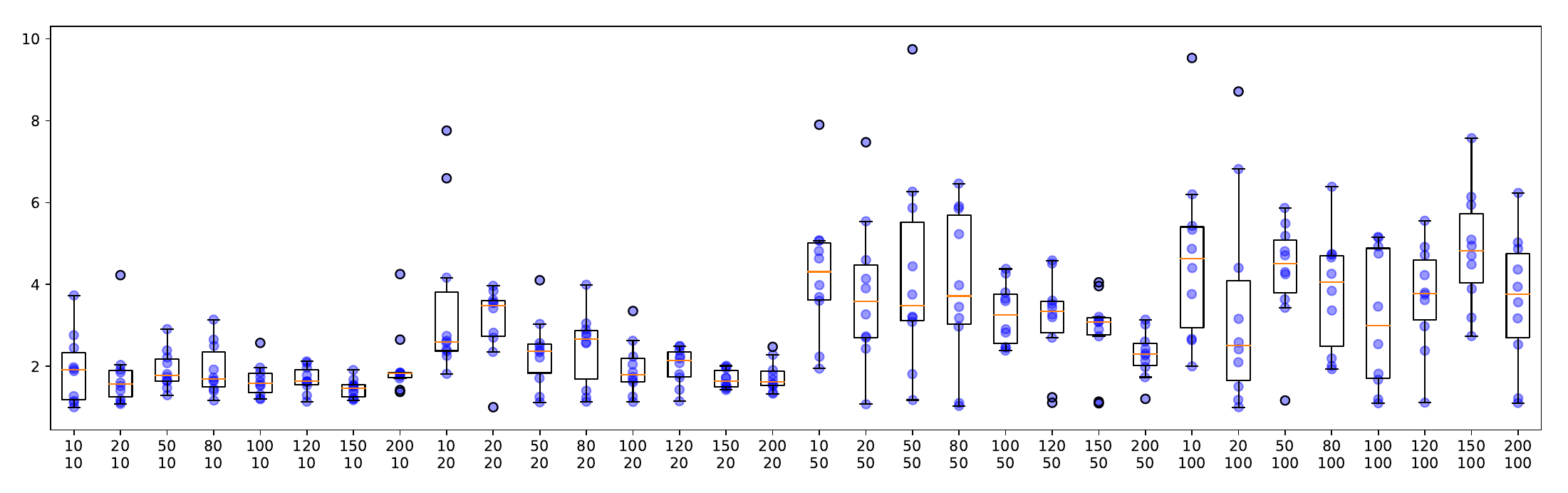}
    
    (b) Scheduling via applying \cref{thm:assignment-to-schedule} to the produced assignments.
    \caption{Empirical approximation ratio for \textsc{Budget-Additive} speed functions. Each point in the $x$-axis is labeled as $\genfrac{}{}{0pt}{}{n}{m}$, where $n$ and $m$ is the number of jobs and machines, respectively.
    }
    \label{fig:plotsBudgeted}
\end{figure}

\subsection{Empirical Evaluation of \texorpdfstring{\Cref{thm:assignment-to-schedule}}{Theorem \ref{thm:assignment-to-schedule}} on Optimal Assignments}

In addition to the empirical evaluation of our heuristic, we further experimentally study the empirical loss due to the application of \cref{thm:assignment-to-schedule} on optimal assignments. In this case, for small enough number of jobs and machines, denoted by $n$ and $m$, respectively, an optimal assignment is computed efficiently via the MILP formulation described above. In \cref{fig:plotsNice}, we report the ratio between the makespan of the scheduled produced by applying \cref{thm:assignment-to-schedule} on the optimal assignment and the makespan of this assignment, for each of the speed functions defined above. We remark that for the case of \textsc{Matroid-Matching}, we slightly modify the number of slots and groups to be $k = \lfloor m \rfloor + 1$ and $b = \lfloor m \rfloor + 1$, respectively. The reason behind this modification is that, for the original values of $k,b$ and for small number of machines, the optimal assignment is always a valid schedule where each job is assigned to a single machine. 

As we can see, the approximation loss for all instances lies between $1$ and $2$, which is smaller than the theoretically proven guarantee of $\frac{2e}{e-1} \approx 3.163$. In fact, the average empirical approximation ratio for all the instances is approximately $1.179$. An interesting observation is that in many examples, the reported ratio is exactly $1$. This implies that in the optimal assignment, either all jobs are assigned to machines as singletons (and, thus, \cref{thm:assignment-to-schedule} has no effect), or that the assignment is already a schedule of nice structure (where each machine contributes to the execution of at most one parallel job) which happens to be maintained by the rounding procedure of \cref{thm:assignment-to-schedule}. Finally, we can see that for any speed function the approximation loss decreases as the ratio of jobs to machines increases. This phenomenon matches our intuition that, for a fixed number of machines and as the number of jobs increases, more and more jobs are already assigned as singletons in an optimal assignment. The reasoning behind this intuition is that this type of allocation minimizes the total work (number of machines times processing time) invested in each job, and is thus, on a large scale, the most efficient way of using the machines.

\begin{figure}
\centering
    \includegraphics[width=\textwidth]{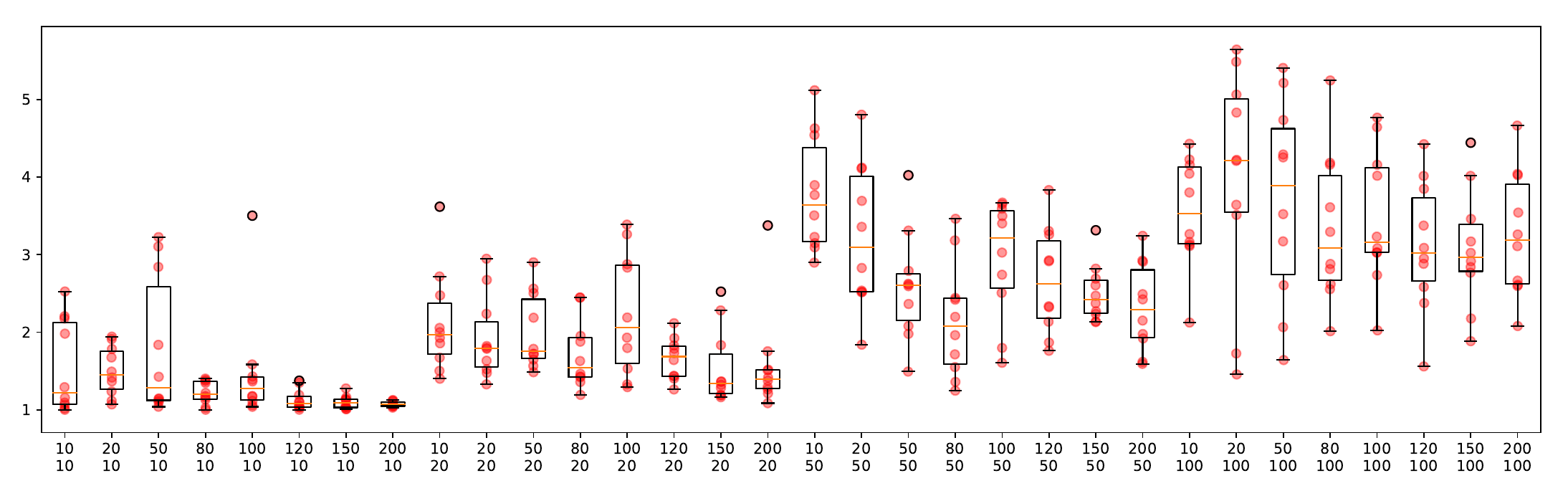}
    
    (a) Assignment using the heuristic of \cref{sec:submodular}.
    \includegraphics[width=\textwidth]{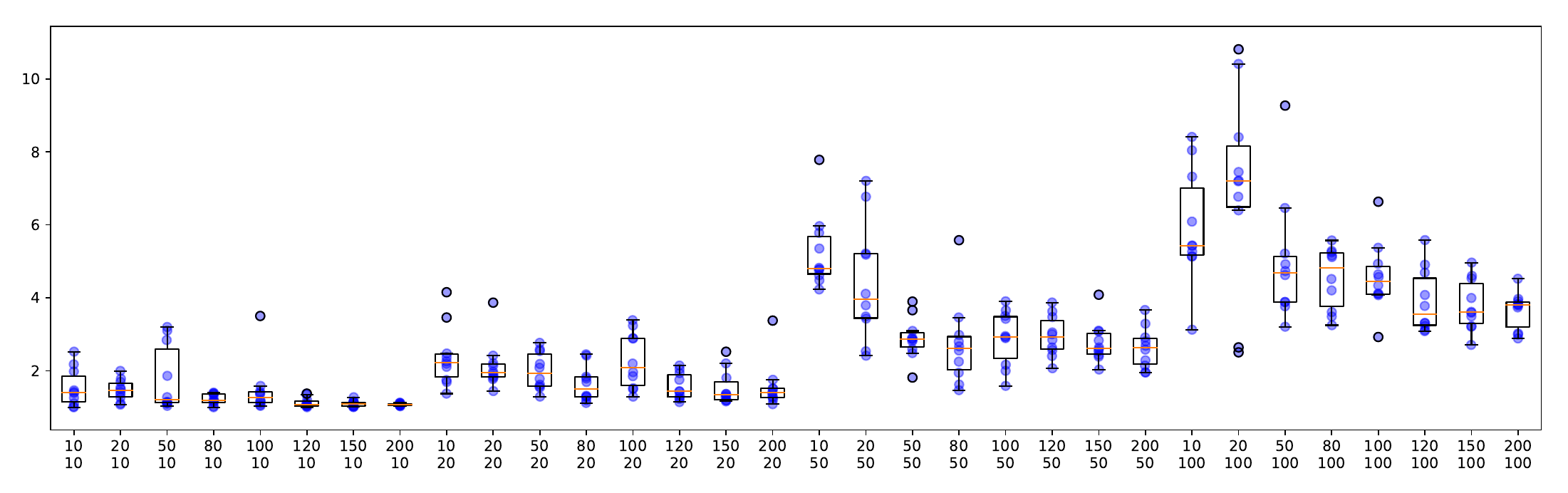}
    
    (b) Scheduling via applying \cref{thm:assignment-to-schedule} to the produced assignments.
    \caption{Empirical approximation ratio for \textsc{Matroid-Matching} speed functions. Each point in the $x$-axis is labeled as $\genfrac{}{}{0pt}{}{n}{m}$, where $n$ and $m$ is the number of jobs and machines, respectively.}
    \label{fig:plotsMatching}
\end{figure}

\begin{figure}[ht]
\centering
\begin{minipage}{.33\textwidth}
  \centering
  \includegraphics[width=\linewidth]{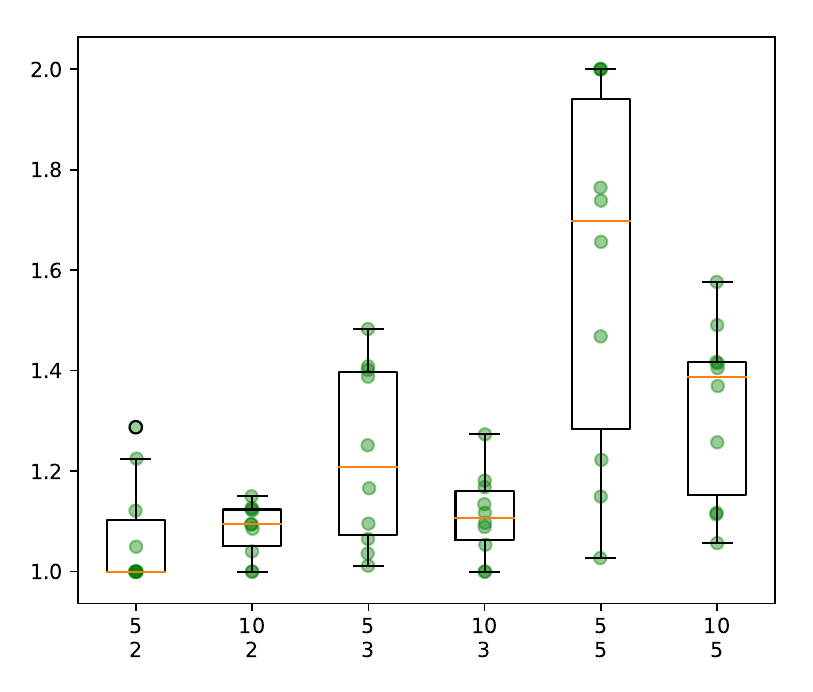}
  
  (a)
  
\end{minipage}%
\begin{minipage}{.33\textwidth}
  \centering
  \includegraphics[width=\linewidth]{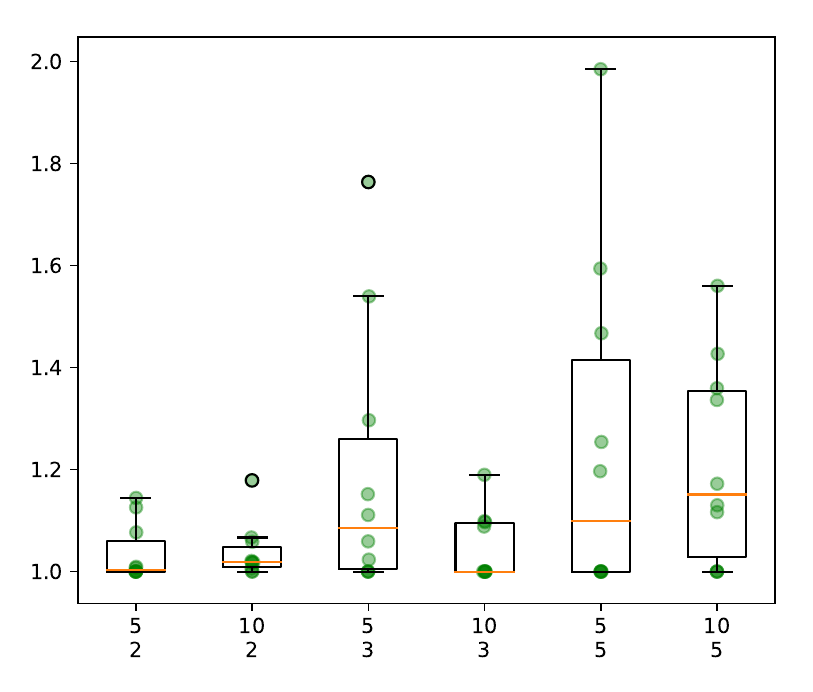}
  
  (b)
  
\end{minipage}
\begin{minipage}{.33\textwidth}
  \centering
  \includegraphics[width=\linewidth]{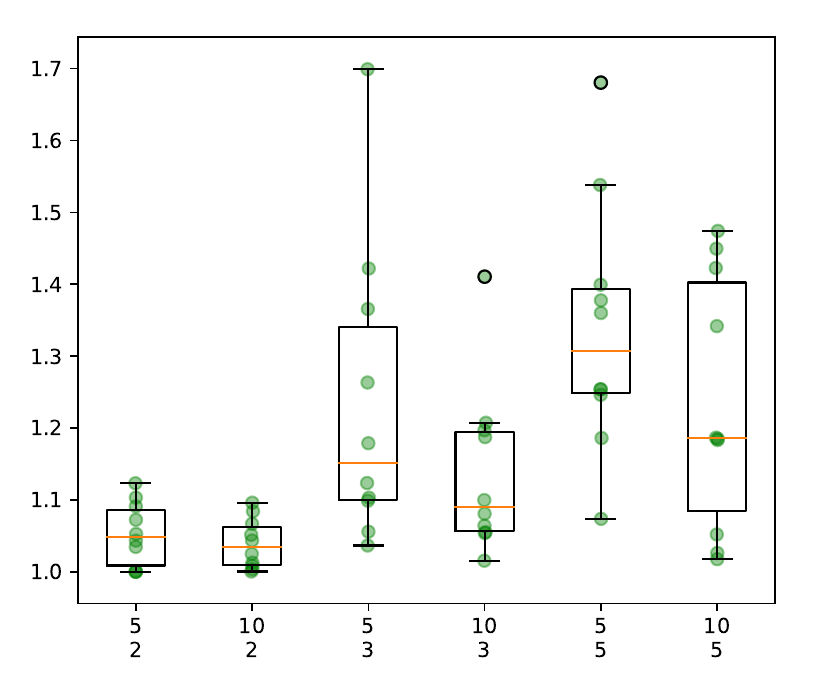}
  
  (c)
  
\end{minipage}
  \caption{Empirical approximation loss due to \Cref{thm:assignment-to-schedule} for (a) \textsc{Coverage}, (b) \textsc{Budget-Additive}, and (c) \textsc{Matroid-Matching} speed functions. Each point in the $x$-axis is labeled as $\genfrac{}{}{0pt}{}{n}{m}$, where $n$ and $m$ is the number of jobs and machines, respectively.}
  \label{fig:plotsNice}
\end{figure}

\bibliographystyle{plainnat}
\bibliography{ref}

\newpage
\appendix

\crefalias{section}{appsec}
\crefalias{subsection}{appsec}

\section{Appendix to \texorpdfstring{\cref{sec:intro}}{Section \ref{sec:intro}}}

\subsection{An Example with Large \textsc{Assignment}-\textsc{Schedule} Gap}
\label{app:assignment-scheduling-gap}
Consider $n$ jobs and $\frac{n \cdot (n - 1)}{2}$ machines, where each machine contributes to the execution of exactly two distinct jobs, and each job requires a specific subset of $n-1$ machines in order to be executed with processing time $1$ (otherwise the job cannot be executed, i.e., $f_j(S) = \infty$ for all jobs $j \in J$ and all machine sets $S \subseteq M$ such that do not contain all $n-1$ machines required for job $j$). In order to visualize the allocation of machines to jobs, one can think of a clique graph $K_{n}$ with $n$ nodes, where each node represents a distinct job and every edge corresponds to a machine that contributes to the execution of both its adjacent job nodes.

In the only feasible assignment for the above instance, the load of each machine is exactly $2$, as each machine processes exactly two jobs of processing time $1$. On the other hand, it is not hard to verify that the makespan of any feasible schedule for the same assignment cannot be less than $n$. Indeed, by construction of our instance, every two jobs share exactly one machine (i.e., the one corresponding to their common edge in the clique) and, thus, their executions cannot be overlapping. Hence, since jobs have to be processed sequentially, each with processing time $1$, we conclude that any feasible schedule has makespan at least $n$.

\subsection{IP Formulation for \textsc{Assignment}}
\label{app:assignment-ip}
The following IP is an exact formulation for the \textsc{Assignment} problem, with binary variables $x(S, j)$ for every $j \in J$ and every $S \subseteq M$, indicating that job $j$ should be processed on machine set $S$.

\begin{alignat}{3} 
\minimize \ &&  C  &\tag{IP} \label{IPassignment}\\
\subjectto \ && \sum_{S \subseteq M : S \neq \emptyset} x(S, j) & \ = \ 1 & \quad \forall j \in J \notag\\
&& \sum_{j \in J} \sum_{S \subseteq M : i \in S} \frac{1}{g_j(S)} x(S, j) & \ \leq \ C & \quad \forall i \in M \notag\\
&& x  & \ \in\{0,1\}^{|M| \times |J|} \  \notag
\end{alignat}

\subsection{Approximation Algorithms for Scheduling}
\label{app:approximation}
\paragraph{Relaxed decision procedures.} 
Many approximation algorithms for makespan minimization rely on the construction of a {\em $\rho$-relaxed decision procedure} \citep{LST90}, namely, a polynomial-time routine which, given a value $C$, either correctly asserts that there exists no assignment (resp., schedule) with makespan less than $C$, or returns a feasible assignment (resp., schedule) with makespan at most $\rho C$. It is well-known that access to such a routine immediately implies a polynomial-time $\rho$-approximation algorithm for the underlying problem of makespan minimization. Indeed, this can be achieved by binary search over the possible range of $C$, given that lower and upper bounds on the optimal makespan can be computed. Since all the algorithms presented in this work follow the above regime, we focus on constructing such procedures.

\paragraph{Rounding theorem for scheduling unrelated machines} 
A recurrent theme in our analysis it the use of the following {\em rounding theorem} due to \cite{LST90}. Let $M$ be a set of jobs and $J$ be a set of machines, where each job must be assigned to a single machine (that is, jobs are non-malleable). Let $p_{ij} \in \mathbb{R}_{\geq 0}$ be the {\em unrelated} processing time for assigning job $j$ to machine $i$. 

For any $C$, we consider the following feasibility LP:
\begin{alignat*}{3}
&& \sum_{i \in M} x_{ij} & \ = \ 1 & \quad \forall j \in J \\
&& \sum_{j \in J} p_{ij} x_{ij} & \ \leq \ C & \quad \forall\; i \in M\\
&& x_{ij} & \ = \ 0, \text{ if }p_{ij} > C  \ & \quad \forall j \in J \  \forall i \in M\\
&& x & \ \geq \ 0 &
\end{alignat*}

\begin{theorem}[Rounding theorem \citep{LST90}] \label{thm:roundingLST}
Any feasible (fractional) extreme point solution to the above LP can be rounded into an integral solution of makespan at most $2C$. 
\end{theorem}

We now briefly discuss the main idea behind the above result. Given any feasible extreme point solution $x$ to the above LP, we can construct the {\em support graph} $G$ as follows: We consider a node $u_j$ for each job $j$ and a node $v_i$ for each machine $i$. For each $x_{ij} >0 $, we define an edge $\{u_j, v_i\}$. It can be proved that $G$ is a {\em pseudoforest}, in the sense that all its connected components are either trees, or contain at most one cycle. Hence, the edges of $G$ (and of any pseudoforest) can be oriented in a way that each node has an in-degree of at most one. Given such an orientation, the rounding proceeds in two steps: First, all integral assignments are maintained and the involved edges and job-nodes are removed from $G$. Then, every fractionally assigned job is assigned to any of each children machines. Note that each fractionally assigned job always has at least one child machine by construction of the orientation and by the constraints of the LP. By feasibility of the LP, the load of each machine due to the integrally assigned jobs is at most $C$. Further, due to the orientation, each machine is assigned at most one job from the assignment graph and this job has processing time $p_{ij} \leq C$. Thus, the makespan of the produced schedule is at most $2C$.

\section{Appendix to \texorpdfstring{\cref{sec:transformation}}{Section \ref{sec:transformation}}}

\subsection{Transformation for Subadditive Processing Speeds}
\label{app:subadditive-transformation}

As we already remark in \cref{sec:transformation}, the following approximation result is well-known for subadditive functions: 

\begin{lemma}[\citeauthor{Bhawalkar2011WelfareGF}, \citeyear{Bhawalkar2011WelfareGF}] \label{lem:subadditivetransformation}
Let $g: 2^M \rightarrow \mathbb{R}_{\geq 0}$ be a subadditive function. There exists a fractionally subadditive function $h: 2^M \rightarrow \mathbb{R}_{\geq 0}$, such that for any $S \subset M$ it holds
$$
g(S) \geq h(S) \geq \frac{1}{\ln(|M|)} g(S).
$$
\end{lemma}

For completeness, we now provide a transformation due to \cite{Bhawalkar2011WelfareGF} that achieves the above guarantee: 

Let $S \subseteq M$ be the queried subset. We set $T = \emptyset$ and $\beta(i) = 0$ for each $i \in M$. While $T \neq S$, we repeat:

\begin{enumerate}
\item Compute $A = \arg\min_{A' \subseteq S } \frac{g(A')}{|A' \setminus T|}$.
\item For each $i \in A \setminus T$ (in lexicographic order), set $\beta(i) = \frac{g(A)}{|A \setminus T| \cdot \ln(|M|)}$. 
\item Set $T = T \cup A$. 
\end{enumerate}

At the end of the above process, we return $h(S) = \sum_{i \in S} \beta(i)$.

As proved bu \citet{Bhawalkar2011WelfareGF}, the function $h$ which results from the above transformation is fractionally subadditive and satisfies the guarantees of \cref{lem:subadditivetransformation}. 

\restatesubadditive*

\begin{proof}
Consider any assignment $\mathbf{S}$ with $L(\mathbf{S}) = C$. Consider a new instance of the \textsc{Scheduling} problem, in which the processing speed function $g_j$ of each job $j$ is  replaced by its fractionally subadditive approximation $h_j$ as given by \cref{lem:subadditivetransformation}.
Note that for the assignment $\mathbf{S}$ the maximum machine load of any machine with respect to processing speeds $h_j$ is at most $\ln(|M|) \cdot C$. 
For this new instance, moreover, \cref{thm:assignment-to-schedule-XOS} guarantees the existence of a well-structured assignment $\mathbf{S}'$ of maximum machine load at most $\frac{2e}{e-1} \ln(|M|) \cdot C$ with respect to $h_j$.
Because $g_j(S) \leq h_j(S)$ for all $j \in J$ and all $S \subseteq M$, we conclude that  $L(\mathcal{S}) \leq \frac{2e}{e-1} \ln(|M|) \cdot C$ when measuring machine loads for $\mathbf{S}'$ with respect to the original processing speeds $g_j$ again.
\end{proof}

\subsection{Existence of \texorpdfstring{$\Delta$}{Delta} (Proof of \texorpdfstring{\cref{lem:delta}}{Lemma~\ref{lem:delta}})}
\label{app:delta-existence}

\restatedelta*

As noted in \cref{sec:transformation}, the lemma follows from the Bondareva-Shapley theorem~\citep{bondareva1963some}. For completeness, we give a direct proof using LP duality below.

\begin{proof}
Let $j \in J$. Consider the following pair of primal and dual LP.

\begin{minipage}{.35\textwidth} %
  \begin{alignat*}{3}
\quad \maximize \quad && \sum_{i \in S_j} &\Delta_{ij}\\
\subjectto \quad && \sum_{i \in T} \Delta_{ij} & \ \leq \ g_j(T)  \quad &\forall\, T \subseteq S_j \\
&& \Delta & \ \geq \ 0
\end{alignat*}
\end{minipage}
\begin{minipage}{.35\textwidth}
  \begin{alignat*}{3}
\quad \minimize \quad && \sum_{T \subseteq S_j} &\alpha(T) \cdot g_j(T)\\
\subjectto \quad && \sum_{T \subseteq S_j : i \in T} \alpha(T) & \ \geq \ 1 \quad & \forall\, i \in S_j \\
&& \alpha & \ \geq \ 0
\end{alignat*}
\end{minipage}

Note that both LPs have feasible solutions and let $\Delta$ and $\alpha$ be optimal solutions to the above LPs. By optimality and feasibility of $\alpha$, it is not hard to verify that $0 \leq \alpha(T) \leq 1$ for each $T \subseteq S_j$. 
By strong duality and the fact that $g_j$ is fractionally subadditive for each $j \in J$, we get
\begin{align*}
    \sum_{i \in S_j} \Delta_{ij} = \sum_{T \subseteq S_j} \alpha(T) g_j(T) \geq g_j(S_j).
\end{align*}
Therefore, the values $\Delta_{ij}$ obtained from optimal solutions to the primal LP above, for each $j \in J$, fulfill the conditions of the lemma.
\end{proof}

\subsection{Bounding the Maximum Load in the Reduced Assignment (\texorpdfstring{\cref{lem:xos-final-guarantee}}{Lemma~\ref{lem:xos-final-guarantee}})}
\label{app:reduced-assignment}

For convenience, we recall the definitions of $J^1_i := \{j \in J \st i_{j} = i, x'_{ij} > 1/2\}$ as the set of jobs assigned to any machine $i \in M$ via a Type~1 assignment and
$\ell_i := \sum_{j \in J^1_i} \frac{1}{\Delta_{ij}} x'_{ij}$ as the (fractional) load of $i$ due to such Type~1 assignments in the extreme point solution $x'$ to \eqref{lp:delta}. 
We further recall that we defined $T_j(\theta) := \{i \in \bar{S}_j : \ell_i \leq \theta\}$ for $j \in J$ and that $\theta^*_j$ was chosen a minimizer of $\max_{i \in T_j(\theta)} 2\ell_i + f_j(T_j(\theta))$.

\restatefinalguarantee*

\begin{proof}
The fact that $\mathbf{T}$ is well-structured follows immediately from the fact that $\mathbf{S}'$ is well-structured and that $T_j \subseteq S'_j$ for all $j \in J$.
In particular, for every machine $i \in M$, the set of jobs processed by $i$ consists of the jobs $j \in J^1_i$ assigned to $i$ via a Type 1 assignment and potentially a single job $j' \in J$ with $i \in T_{j'}(\theta^*_{j'})$ with $x'_{i_{j'}i} \leq 1/2$ assigned to $i$ via a Type 2 assignment.
Recall that in the proof of \cref{lem:type-1-assignment}, we showed that 
\begin{align}
    \sum_{j \in J^1_i} f_j(\{i\}) \leq 2\ell_i \leq 2C. \label{eq:bound-ell}
\end{align}
Thus, if a machine $i \in M$ only processes jobs assigned to it via Type 1 assignments, then its load in $\mathbf{T}$ is at most $2C \leq \tfrac{2e}{e-1}C$.
It thus remains to bound the load of machines $i \in M$ for which there exist some Type~2 job $j \in J$ with $i \in T_{j}(\theta^*_{j})$ and $x'_{i_{j}i} \leq 1/2$.

Now fix any such Type~2 job $j \in J$ and note that the load of any machine $i \in T_j = T_j(\theta^*)$ is $f_j(T_j(\theta^*)) + \sum_{j \in J^1_{i}} f_j(\{i\}) \leq f_j(T_j(\theta^*)) + 2\ell_i$ by \eqref{eq:bound-ell}.
The maximum load over all machines in $T_j$ is therefore bounded by
\begin{align*}
     \max_{i \in T_j(\theta^*_j)} \{2 \ell_i\} + f_j(T_j(\theta^*_j)) \leq  2\theta^*_j + f_j(T_j(\theta^*_j)),
\end{align*}
where the inequality follows by definition of $T_j(\theta)$. 

As $\theta^*_j$ was chosen to minimize $2\theta + f_j(T_j(\theta))$, it suffices to show that there exists some $\theta > 0$ such that $2 \theta + f_j(T_j(\theta)) \leq \alpha C$, for $\alpha:= \frac{2e}{e-1}$.
Define $h_j(\theta):= \sum_{i \in T_j(\theta)} \Delta_{ij}$ and notice that
\begin{align*}
    \int_{\theta=0}^{C} h_j(\theta) d\theta = \sum_{i \in \bar{S}_j} \Delta_{ij} \left( C - \ell_i \right).
\end{align*}
Note that $\ell_i + \frac{1}{\Delta_{ij}}x_{ij} \leq C$ for each $i \in \bar{S}_j$, by feasibility of the LP solution and the definition of $\ell_i$. 
By summing over all $i \in \bar{S}_j$ and using the fact that $\sum_{i \in \bar{S}_j} x_{ij} \geq \frac{1}{2}$, we obtain $\frac{1}{2} \leq \sum_{i \in \bar{S}_j} \Delta_{ij} \left(C - \ell_i \right)$ and, thus, $\frac{1}{2} \leq \int_{\theta=0}^C h_j(\theta) d\theta$. 

Now assume by contradiction that $2 \theta + f_j(T_j(\theta)) > \alpha C$ for all $\theta > 0$, Not that this is equivalent to $g_j(T_j(\theta)) < \frac{1}{\alpha C - 2\theta}$ and hence we obtain
\begin{align*}
\frac{1}{\alpha C - 2\theta} > g_j(T_j(\theta)) \geq  \sum_{i \in T_j(\theta)} \Delta_{ij} = h_j(\theta),
\end{align*}
where the last inequality follows 
by \cref{lem:delta} applied to $T = T_j(\theta)$. By combining the above bounds, we obtain
\begin{align*}
    \frac{1}{2} \leq \int_{\theta=0}^C h_j(\theta) d\theta < \int_{\theta=0}^C \frac{1}{\alpha C - 2\theta} d\theta.
\end{align*}
Multiplying both sides by $2$, resolving the integral and using that $\alpha = \frac{2e}{e-1}$ we obtain the contradiction
\begin{align*}
    1 < \int_{\omega=\alpha-2}^\alpha \frac{1}{\omega} d\omega = \ln\left(\frac{\alpha}{\alpha-2} \right)=1.
\end{align*}
This shows that $2 \theta^*_j + f_j(T_j(\theta^*_j)) \leq \frac{2e}{e-1} C$ and thus completes the proof of the lemma.
\end{proof}

\section{Appendix to \texorpdfstring{\cref{sec:matroids}}{Section \ref{sec:matroids}}}

\subsection{Orienting the Graph in Step~3b}
\label{app:graph-orientation}

\begin{lemma}[\citeauthor{hakimi65}, \citeyear{hakimi65}]\label{lem:matroid-step3-feasible}
 Let $G = (V, E)$ be an undirected graph. For $v \in V$, let $d(v)$ be the degree of node $v$ in $G$. There is an orientation of $G$ such that the in-degree of each node $v \in V$ is at least $\lfloor d(v) / 2 \rfloor$.
\end{lemma}

For completeness, we provide a simple proof of the lemma that also explains how such an orientation can be constructed in polynomial time. 

\begin{proof}
Assume w.l.o.g. that $G$ is a connected graph. Clearly, if $G$ is Eulerian, the orientation can be computed simply by finding an Euler tour and orienting the edges of the cycle in the direction of the tour. If $G$ is not Eulerian, then by the handshaking lemma, we know that the number of odd-degree nodes must be even. In that case, we can compute an arbitrary matching between the odd-degree nodes (i.e., we group these nodes in pairs), and add artificial edges between matched nodes, so that the graph becomes Eulerian (since the degree of all nodes is now even when including the artificial edges). Now, if we compute an Euler tour in the above graph (including the artificial edges) and orient it in a consistent way, the lemma follows simply by dropping the artificial edges. \end{proof}

\section{Appendix to \texorpdfstring{\cref{sec:submodular}}{Section }}

\subsection{Hardness of Approximation for Fractionally Subadditive Processing Speeds (\texorpdfstring{\cref{thm:hardness}}{Theorem~\ref{thm:hardness}})}
\label{app:hardness}

We prove the inapproximability result for \textsc{Assignment} with fractionally subadditive processing speeds given by a value oracle. The proof follows from adapting a construction by \cite{feige2009maximizing} that shows a similar inapproximability for a welfare maximization problem.

\restatehardness*

\begin{proof}
We consider an instance of $n \geq 2$ jobs and $m = n^2$ machines, where the set of machines is partitioned into $n$ sets $T_1, \ldots, T_n$ of size $n$ each. 
For each job $j$ and set of machines $S$, let 
$$f_j(S) := \frac{n}{\max\{ |S| + n^{5/3}, n \cdot |S \cap T_j| \}}.$$
It is easy to verify that the processing speeds $g_j(S) := {1}/{f_j(S)}$ are fractionally subadditive for every job $j$, and that the optimal makespan in the above instance is $\opt = 1/n$ (i.e., when each job $j$ is processed just on its corresponding set $T_j$).

The partition $T_1, \ldots, T_n$ is unknown to the decision maker, who can only access the values of $f_j$ via a value oracle.
By the argumentation of \citeauthor{feige2009maximizing} (\citeyear{feige2009maximizing}, Section 1.2), however, the decision maker cannot find any set $S$ with
$|S \cap T_j| \geq \frac{|S|}{n} + n^{{2}/{3}}$ 
within a polynomial number of value queries, since the probability of choosing any such set $S$ when selecting a subset of $M$ uniformly at random is smaller than the inverse of any polynomial in $n$. 
Hence, we can assume that any algorithm using only a polynomial number of value queries will return an assignment $\mathbf{S}$ with $f_j(S_j) := \frac{n}{|S_j| + n^{5/3}}$ for all $j \in J$.

Now, if there is even a single job $j$ with $|S_j| \leq n^{5/3}$, then 
$f_j(S_j) \geq \tfrac{1}{2} \cdot n^{-2/3} = \Omega(n^{1/3}) \cdot \opt$. On the other hand, if $|S_j| \geq n^{5/3}$ for all $j \in J$, then by a volume argument there will be at least one machine running at least $n^{2/3}$ jobs, and each of these jobs takes at least time $\Omega(1/n)$. Thus our makespan is off from $\opt$ by at least a factor of $\Omega(n^{2/3})$.
\end{proof}

\subsection{Proof of \texorpdfstring{\cref{lem:submodular:nosplitting}}{Lemma~\ref{lem:submodular:nosplitting}}}
\label{app:submodular-gap}

\restatenosplitting*

\begin{proof}
Let $J_1 \subseteq J$ be the subset of $|J_1|$ jobs that are assigned without splitting in the given partial assignment and let $i_j \in M$ be the machine where each job $j \in J_1$ is assigned to. We consider the following LP: 
\begin{alignat}{3}
\maximize \ &&\sum_{i \in M} \sum_{j \in J} & x_{ij} \tag{\HLP} \label{HLP}\\
\subjectto \ &&\sum_{i \in M} x_{ij} & \ \leq \ 1 & \quad \forall j \in J \notag \\
&& \sum_{j \in J} \frac{1}{g_j(\{i\})} x_{ij} & \ \leq \ C' & \quad  \forall i \in M \notag \\ 
&& x_{ij} & \ \equiv \ 0 & \quad \forall i \in M, j \in J\text{ with }f_j(\{i\})>C' \notag \\
&& x & \ \geq \ 0 \notag 
\end{alignat}

By setting $x_{ij} = 1$ for each $j \in J_1$ and $i = i_j$, and $x_{ij} = 0$, otherwise, we can easily verify that all the constraints of \eqref{HLP} are trivially satisfied. Further, for the same assignment, the objective becomes $\sum_{i \in M} \sum_{j \in J} x_{ij} = |J_1|$. Hence, we conclude that \eqref{HLP} is feasible and its optimal value is at least $|J_1|$. Let $x$ be an optimal extreme point solution of \eqref{HLP}, which satisfies $\sum_{i \in M} \sum_{j \in J} x_{ij} \geq |J_1|$. Given that for each job $j$ it holds $\sum_{i \in M} x_{ij} \leq 1$, it follows that at least $|J_1|$ jobs appear in the support of $x$ (that is, $\sum_{i \in M} x_{ij} > 0$ for at least $|J_1|$ jobs). Thus, it suffices to prove that we can round $x$ into a feasible schedule in a way such that any job that appears in its support is integrally assigned to a machine. 

By working along the lines of \cite{LST90} (see \cref{app:approximation}), we can construct an assignment graph based on the support of $x$. It can be verified that, if $x$ is an extreme point solution of \eqref{HLP}, then the graph is a pseudoforest \citep{shmoys1993approximation}. Therefore, by using the rounding algorithm of \cite{LST90} we can schedule all the jobs in the support of $x$ within makespan at most $2C'$ and, thus, the desired lemma follows.
\end{proof}

\end{document}